\documentclass[sigconf,screen,nonacm]{acmart}

\usepackage{balance}
\usepackage{amsmath,amssymb,amsfonts}
\usepackage{algorithmic}
\usepackage{graphicx}
\usepackage{textcomp}
\usepackage{xcolor}
\usepackage{alltt}
\usepackage{microtype}

\usepackage{stmaryrd}
\usepackage{xspace}
\usepackage{booktabs}
\usepackage{mdframed}
\usepackage{enumitem}
\usepackage{thm-restate}

\usepackage{url}

\setlist{leftmargin=*}

\newtheorem{theorem}{Theorem}
\newtheorem{property}{Property}
\newtheorem{assumption}{Assumption}
\newtheorem*{ih}{IH}
\theoremstyle{remark}

\usepackage[position=top]{subfig}

\usepackage{todonotes}

\newcommand{\ie}{\textit{i.e.,}\xspace}
\newcommand{\eg}{\textit{e.g.,}\xspace}
\newcommand{\etal}{\textit{et al.}\xspace}

\usepackage{listings}
\usepackage{xcolor}
\usepackage{textcomp} 

\usepackage{cleveref}
\crefformat{section}{\S#2#1#3}
\crefname{figure}{Figure}{Figures}
\crefname{table}{Table}{Tables}
\crefname{listing}{Listing}{Listings}
\crefname{assumption}{Assumption}{Assumptions}

\newcommand{\code}[1]{\texttt{{\small #1}}}

\colorlet{punct}{red!60!black}
\definecolor{background}{HTML}{EEEEEE}
\definecolor{delim}{RGB}{20,105,176}
\colorlet{numb}{magenta!60!black}
\lstdefinelanguage{json}{
    basicstyle=\normalfont\ttfamily,
    morecomment=[f][\color{green!60!black}][0]{\#},
    numbers=left,
    numberstyle=\scriptsize,
    stepnumber=1,
    numbersep=8pt,
    showstringspaces=false,
    breaklines=true,
    frame=lines,
    backgroundcolor=\color{background},
    literate=
     *{0}{{{\color{numb}0}}}{1}
      {1}{{{\color{numb}1}}}{1}
      {2}{{{\color{numb}2}}}{1}
      {3}{{{\color{numb}3}}}{1}
      {4}{{{\color{numb}4}}}{1}
      {5}{{{\color{numb}5}}}{1}
      {6}{{{\color{numb}6}}}{1}
      {7}{{{\color{numb}7}}}{1}
      {8}{{{\color{numb}8}}}{1}
      {9}{{{\color{numb}9}}}{1}
      {:}{{{\color{punct}{:}}}}{1}
      {,}{{{\color{punct}{,}}}}{1}
      {\{}{{{\color{delim}{\{}}}}{1}
      {\}}{{{\color{delim}{\}}}}}{1}
      {[}{{{\color{delim}{[}}}}{1}
      {]}{{{\color{delim}{]}}}}{1},
}

\definecolor{keyword_color}{rgb}{0.8431372549019608, 0.0, 0.023529411764705882}
\definecolor{identifier_color}{rgb}{0.8980392156862745, 0.36470588235294116, 0.0}
\definecolor{string_color}{rgb}{0.1607843137254902, 0.1843137254901961, 0.40784313725490196}
\definecolor{number_color}{rgb}{0.027450980392156862, 0.396078431372549, 0.796078431372549}
\definecolor{annotation_color}{rgb}{0.43529411764705883, 0.23529411764705882, 0.7686274509803922}
\definecolor{grey}{rgb}{0.6, 0.6, 0.6}

\makeatletter
\lst@AddToHook{OnEmptyLine}{\vspace{-0.75\baselineskip}}
\newcommand*\idstyle{%
        \expandafter\id@style\the\lst@token\relax
}
\def\id@style#1#2\relax{%
        \ifcat#1\relax\else
                \ifnum`#1=\uccode`#1%
                        \color{number_color}
                \else
                        \color{identifier_color}
                \fi
        \fi
}
\makeatother

\lstdefinelanguage{graphql}{
  identifierstyle=\idstyle,
  delim=[s][\color{string_color}]{"}{"},%
  morecomment=[f][\color{green!60!black}][0]{\#},
  literate=
     *{0}{{{\color{number_color}0}}}{1}
      {1}{{{\color{number_color}1}}}{1}
      {2}{{{\color{number_color}2}}}{1}
      {3}{{{\color{number_color}3}}}{1}
      {4}{{{\color{number_color}4}}}{1}
      {5}{{{\color{number_color}5}}}{1}
      {6}{{{\color{number_color}6}}}{1}
      {7}{{{\color{number_color}7}}}{1}
      {8}{{{\color{number_color}8}}}{1}
      {9}{{{\color{number_color}9}}}{1}
      {query\ }{{{\color{keyword_color}{query }}}}{1}
      {mutation\ }{{{\color{keyword_color}{mutation }}}}{1}
      {schema}{{{\color{keyword_color}{schema }}}}{1}
      {type\ }{{{\color{keyword_color}{type }}}}{1}
      {input\ }{{{\color{keyword_color}{input }}}}{1}
      {!}{{{\color{keyword_color}{!}}}}{1}
      {@deprecated}{{{\color{annotation_color}{@deprecated}}}}{1}
      {...\ on}{{{\color{keyword_color}{\texttt{... on }}}}}{1}
      ,
}

\lstnewenvironment{lstgql}%
{\lstset{
  language=graphql,
  basicstyle=\footnotesize\ttfamily,
  columns=fullflexible,
  xleftmargin=0em,
  belowskip=0em,
  aboveskip=0em,
  morecomment=[l][\color{olive}]{\#},
}}
{}
\def\gql{\lstinline[language=graphql, basicstyle=\small\ttfamily]}

\lstnewenvironment{lstgqlsmall}%
{\lstset{
  language=graphql,
  basicstyle=\scriptsize\ttfamily,
  columns=fullflexible,
  numbers=none,
  xleftmargin=3em,
}}
{}

\newcommand{\MyParagraph}[1]{\smallskip \textit{#1}\xspace}

\newif\ifextended
\extendedfalse

\begin{document}

\newcommand{\MyTitle}{}
\renewcommand{\MyTitle}{Query Analysis for GraphQL API Management}
\renewcommand{\MyTitle}{A Principled Approach to GraphQL Query Analysis}
\renewcommand{\MyTitle}{A Principled Approach to GraphQL Query Cost Analysis}
\title{\MyTitle}

\author{Alan Cha}
\affiliation{
  \institution{IBM Research, USA}            
}
\email{alan.cha1@ibm.com}          

\author{Erik Wittern}
\authornote{Most of the work performed while at IBM Research, USA.}          
\affiliation{
  \institution{IBM, Germany}            
}
\email{erik.wittern@ibm.com}          

\author{Guillaume Baudart}
\affiliation{
  \institution{IBM Research, USA}            
}
\email{Guillaume.Baudart@ibm.com}          

\author{James C. Davis}
 \authornote{Most of the work performed while at Virginia Tech.}          
\affiliation{
  \institution{Purdue University, USA}        
}
\email{davisjam@purdue.edu}          

\author{Louis Mandel}
\affiliation{
  \institution{IBM Research, USA}            
}
\email{lmandel@us.ibm.com}          

\author{Jim A. Laredo}
\affiliation{
  \institution{IBM Research, USA}            
}
\email{laredoj@us.ibm.com}          

\begin{abstract}
The landscape of web APIs is evolving to meet new client requirements and to facilitate how providers fulfill them.
A recent web API model is GraphQL, which is both a query language and a runtime.
Using GraphQL, client queries express the data they want to retrieve or mutate, and servers respond with exactly those data or changes.
GraphQL's expressiveness is risky for service providers because clients can succinctly request stupendous amounts of data, and responding to overly complex queries can be costly or disrupt service availability. 
Recent empirical work has shown that many service providers are at risk.
Using traditional API management methods is not sufficient, and practitioners lack principled means of estimating and measuring the cost of the GraphQL queries they~receive. 

In this work, we present a linear-time GraphQL query analysis that can measure the cost of a query without executing it.
Our approach can be applied in a separate API management layer and used with arbitrary GraphQL backends.
In contrast to existing static approaches, our analysis supports common GraphQL conventions that affect query cost, and our analysis is provably correct based on our formal specification of GraphQL semantics.

We demonstrate the potential of our approach using a novel GraphQL query-response corpus for two commercial GraphQL APIs.
Our query analysis consistently obtains upper cost bounds, tight enough relative to the true response sizes to be actionable for service providers.
In contrast, existing static GraphQL query analyses exhibit over-estimates and under-estimates because they fail to support GraphQL conventions.

\end{abstract}



\maketitle


\section{Introduction}
\label{sec:introduction}

Web APIs are the preferred approach to exchange information on the internet.
Requirements to satisfy new client interactions have led to new web API models such as GraphQL~\cite{GraphQLInvention:2015}, a data query language.
A GraphQL service provides a \textit{schema}, defining the data entities and relationships for which clients can \textit{query}.

GraphQL has seen increasing adoption because it offers three advantages over other web API paradigms.
First, GraphQL reduces network traffic and server processing because users can express their data requirements in a single query~\cite{Brito:2019}.
Second, it simplifies API maintenance and evolution by reducing the number of service endpoints~\cite{2015GraphQLMotivation}.
Third, GraphQL is strongly typed, facilitating tooling including
  data mocking~\cite{GraphQL-Faker}, 
  query checking~\cite{GraphiQL},
  and wrappers for existing APIs~\cite{Wittern:2018}.
These benefits have not been lost on service providers~\cite{GraphQLUsers}, with adopters including 
  GitHub~\cite{GitHubGraphQL:2019} and
  Yelp~\cite{YelpGraphQL:2017}.
However, GraphQL can be perilous for service providers.
A worst-case GraphQL query requires the server to perform an exponential amount of work~\cite{Hartig:2018}, with implications for execution cost, pricing models, and denial-of-service~\cite{Crosby:2003}.
This risk is not hypothetical ---
  a majority of GraphQL schemas expose service providers to the risk of high-cost queries~\cite{Wittern:2019}.
As practitioners know~\cite{StackOverflow:2016,Rinquin:2017,Stoiber:2018},
\textit{the fundamental problem for GraphQL API management is the lack of a cheap, accurate way to estimate the cost of a query}.

Existing dynamic and static cost estimates fall short (\cref{sec:bm}).
Dynamic approaches are accurate but impractically expensive~\cite{Hartig:2018,Andersson2018GQLResultSizeCalculation}, relying on interaction with the backend service and assuming specialized backend functionality~\cite{Andersson2018GQLResultSizeCalculation}.
Current static approaches are inaccurate and do not support GraphQL conventions~\cite{graphql-validation-complexity,graphql-query-complexity,graphql-cost-analysis}.

{
\setlength{\belowcaptionskip}{0em} 
\begin{figure*}
  \centering
  \includegraphics[width=7in]{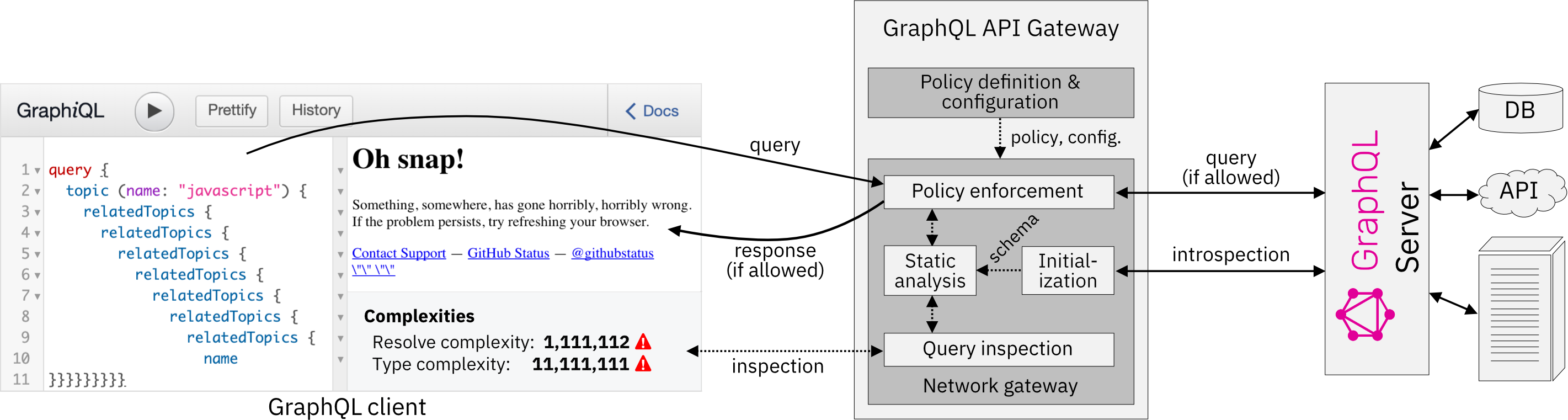}
  \caption{
    Proposed applications of our query analysis. 
    The client's malicious query requests an exponentially large result from GitHub's GraphQL API.
    At the time of our study, GitHub permitted the shown query, but halted its execution after it exceeded a time limit.
    Using our techniques, client-side \textit{query inspection} can provide feedback during composition (see \textit{Complexities} inset).
    Server-side \textit{policy enforcement} can reject queries and update rate limits based on provider-defined policies.
    We disclosed this denial of service vector to GitHub, and it has since been patched (\cref{sec:evaluation_comparison_ClosedSource}).
  }
  \label{fig:architecture}
\end{figure*}
}

We present the first provably correct static query cost analysis for GraphQL.
We begin with a novel formalization of GraphQL queries and semantics (\cref{sec:semantics}).
After extending the formalization with simple \textit{configuration information} to capture common schema conventions,
  we define two \emph{complexity metrics} reflecting server and client costs for a GraphQL query 
  (\cref{sec:background_complexity}).
Then we show how to compute upper bounds for a query's cost according to these metrics (\cref{sec:query_analysis}).
Our analysis takes linear time and space in the size of the query.

Our analysis is accurate and practical (\cref{sec:evaluation}).
We studied 10,000 query-response pairs from two commercial GraphQL APIs.
Unlike existing analyses, our analysis obtains accurate cost bounds even for pathological queries.
With minimal configuration, our bounds are tight enough to be actionable for service providers.
Our analysis is fast enough to be used in existing request-response flows.

This paper makes the following contributions:

\begin{itemize}[topsep={1.5pt},partopsep={1.5pt}]
  \item We give a novel formalization of GraphQL semantics (\cref{sec:semantics}).
  \item We propose two GraphQL query complexity measures (\cref{sec:background_complexity}) that are used to estimate the cost of a query. We then prove a linear-time static analysis to obtain an upper bound for these measures on a given query (\cref{sec:query_analysis}).
  \item We evaluate our analysis against two public GraphQL APIs and show that it is practical, accurate, and fast (\cref{sec:evaluation}).
  We also identify causes of over- and under-estimation in existing static analyses.
  \item We share the first GraphQL query-response corpus~\cite{artifact}: $10,000$ unique query and response pairs from the GitHub and Yelp APIs.
\end{itemize}

We illustrate applications of our analysis by exploiting a flaw in GitHub's static analysis (\cref{fig:architecture}).\footnote{GitHub's API also has a runtime defense, so the risk to their service was minimal.} 
We issued an exponential-time query to GitHub.
GitHub's analysis incorrectly estimated the query's cost, accepted it, and wasted resources until the evaluation timed out. 
Our analysis can help service providers avoid this situation.
Some large queries are accidental, and our measure of \textit{type complexity} would permit clients to understand the potential sizes of responses before issuing queries, reducing accidental service and network costs.
Both type complexity and our measure of \textit{resolve complexity} would permit service providers to understand a query's potential execution cost.
Using these metrics will allow service providers to identify high-cost queries and respond appropriately. 

\section{Background and Motivation}
\label{sec:bm}
In this section, we
  motivate the need for GraphQL query cost analysis (\cref{sec:bm-motivation})
  and then discuss existing query analyses (\cref{sec:bm-existingAnalyses}).

\subsection{Motivation} \label{sec:bm-motivation}

Our work is motivated by two aspects of software engineering practice:
  (1)~the majority of real-world GraphQL schemas expose service providers to high-cost queries,
  and (2)~existing strategies employed by service providers are inadequate.

\paragraph{High-complexity GraphQL schemas are common in practice.} \label{sec:bm-motivation-wittern}

Hartig and P{\'{e}}rez showed that a GraphQL query can yield an exponential amount of data in the size of the query~\cite{Hartig:2018}.
Such a query requests the nested retrieval of the same data, and is only possible
  if the schema defines self-referential relationships (``loops of lists''),
  and
  if the underlying data contains such relationships.
Wittern \etal extended their analysis to identify schemas with polynomial-sized worst-case responses, and analyzed a corpus of GraphQL schemas for these properties~\cite{Wittern:2019}.
In their corpus, they found that
  \textbf{over 80\%} of the commercial or large-scale open-source schemas had exponential worst-case behavior,
  and that
  \textbf{under 40\%} of \textit{all} schemas guaranteed linear-time queries.

\paragraph{Many public GraphQL APIs do not document any query analysis.} 
We manually studied the documentation for the $30$ public APIs listed by \emph{APIs.guru}, a community-maintained listing of GraphQL APIs~\cite{APIsGuru-GraphQLAPIs}.
We used public APIs listed as of February $28^{\text{th}}$, 2020; other GraphQL APIs are unlisted or private~\cite{GraphQLUsers}.
Disturbingly, $25$ APIs ($83\%$) describe neither static nor dynamic query analysis to manage access and prevent misuse.
%
$22$ APIs ($73\%$) make no reference to rate limiting or to preventing malicious or overly complex requests.
Three APIs ($10\%$) perform rate limiting, but only by request \textit{frequency}, ignoring the wide range of query complexities.

A few APIs have incorporated customized query and/or response analysis into their management approach.
Five APIs ($17\%$) describe analyzing GraphQL queries to apply rate limiting based on the estimated or actual cost of a query or response.
GitHub~\cite{GitHubGraphQL:ResourceLimitations:2019}, Shopify~\cite{ShopifyGraphQL:Ratelimits:2019}, and Contentful~\cite{Contentful:Ratelimits:2019} estimate the cost of queries before executing them.
Shopify and Yelp~\cite{Yelp:Ratelimits:2019} update remaining rate limits by analyzing responses, i.e., the actual data sent to clients.
But these approaches have shortcomings that are discussed in~\cref{sec:evaluation_comparison}.

\subsection{Existing GraphQL Query Cost Analyses}
\label{sec:bm-existingAnalyses}
A GraphQL \textit{query cost analysis} measures the cost of a query without fully executing it.
Service providers can use such an analysis to avoid denial of service attacks, as well as for management purposes. 

There are two styles of GraphQL query analysis: \emph{dynamic}~\cite{Hartig:2018} and \emph{static}~\cite{graphql-validation-complexity,graphql-query-complexity,graphql-cost-analysis}.
The dynamic analysis of~\cite{Hartig:2018} considers a query in the context of the data graph on which it will be executed.
Through lightweight query resolution, it steps through a query to determine the number of objects involved.
This cost measure is accurate but expensive to obtain,
  because it incurs additional runtime load and potentially entails engineering costs~\cite{Andersson2018GQLResultSizeCalculation}.\footnote{
    In particular, their analysis repeatedly interacts with the GraphQL backend, and assumes that the backend supports cheap queries for response size.
    This is plausible if the backend is a traditional database, but GraphQL is backend agnostic (\cref{sec:semantics}).
  }


Static analyses~\cite{graphql-validation-complexity,graphql-query-complexity,graphql-cost-analysis} calculate the worst-case query cost supposing a pathological data graph.
Because a static analysis assumes the worst,
  it can efficiently provide an upper bound on a query's cost
  without
  interacting with the backend.
The speed and generality of static query analysis makes this an attractive approach for commercial GraphQL API providers.


Our static approach follows a similar paradigm as existing static analyses but we differ in several ways:
(1)~We provide two distinct definitions of query complexity, which is used to measure query cost;
(2)~Our analysis can be configured to handle common schema conventions to produce better estimates;
(3)~We build our analysis on formal GraphQL semantics and prove the correctness of our query complexity estimates; and
(4)~We perform the first evaluation of such an analysis on real-world APIs.
Overall, our evaluation shows the benefits of a formal and configurable approach, identifying shortcomings in existing static analyses.
\newcommand{\gqmatch}[3]{\ensuremath{\mathit{#1} \, \texttt{:} \mathit{#2} \;\texttt{(}\mathit{#3}\texttt{)}}}
\newcommand{\gqon}[2]{\ensuremath{\textnormal{\texttt{...on}} \; \mathit{#1} \;\texttt{\{} #2 \texttt{\}}}}
\newcommand{\gqconj}[2]{\ensuremath{{#1}\,{#2}}}
\newcommand{\gqnest}[2]{#1 \texttt{\{} #2 \texttt{\}}}

\newcommand{\gsem}[3]{\ensuremath{\llbracket #1 \rrbracket({#2},{\mathit{#3}})}}
\newcommand{\grsize}[4]{\ensuremath{\texttt{qrcx(} #2, \mathit{#1}, {#3},\mathit{#4}\texttt{)}}}
\newcommand{\gtsize}[4]{\ensuremath{\texttt{qtcx(} #2, \mathit{#1}, {#3},\mathit{#4}\texttt{)}}}

\newcommand{\jstypeof}[1]{\ensuremath{\texttt{typeof(} #1 \texttt{)}}}
\newcommand{\jsisarray}[1]{\ensuremath{\texttt{[} #1 \texttt{]}}}
\newcommand{\jsrsize}[2]{\ensuremath{\texttt{rcx(} \mathit{#1}, #2 \texttt{)}}}
\newcommand{\jstsize}[2]{\ensuremath{\texttt{tcx(} \mathit{#1}, #2 \texttt{)}}}
\newcommand{\jsqsize}[1]{\ensuremath{\texttt{size(} #1 \texttt{)}}}
\newcommand{\jslength}[1]{\ensuremath{\texttt{length(} #1 \texttt{)}}}
\newcommand{\jsnull}{\ensuremath{\texttt{null}}}
\newcommand{\jsor}[2]{\ensuremath{#1 \, \texttt{||} \, #2}}
\newcommand{\jsobject}[1]{\ensuremath{\texttt{\{} #1 \texttt{\}}}}
\newcommand{\jsfield}[2]{\ensuremath{\mathit{#1}: #2}}
\newcommand{\jslist}[2]{\ensuremath{\texttt{[}#1 \texttt{,} ... \texttt{,} #2 \texttt{]}}}
\newcommand{\jsmerge}[2]{\ensuremath{\texttt{merge(} #1 \texttt{,} #2 \texttt{)}}}
\newcommand{\jsresolve}[4]{\ensuremath{\texttt{resolve(}\mathit{#1} \texttt{,} \mathit{#2} \texttt{,} \mathit{#3} \texttt{,} \mathit{#4}\texttt{)}}}
\newcommand{\jslimit}[5]{\ensuremath{\texttt{limit(}\mathit{#1} \texttt{,} \mathit{#2} \texttt{,} \mathit{#3} \texttt{,} \mathit{#4} \texttt{,} \mathit{#5}\texttt{)}}}
\newcommand{\jsaccess}[2]{\ensuremath{\mathit{#1}\texttt{.}\mathit{#2}}}
\newcommand{\jsget}[2]{\ensuremath{\mathit{#1}\texttt{[}\mathit{#2}\texttt{]}}}
\newcommand{\jsstring}[1]{\ensuremath{\texttt{"}\mkern-2mu\mathit{#1}\texttt{"}}}

{
\setlength{\belowcaptionskip}{0em} 
\begin{figure*}
  \centering
\begin{minipage}[t]{0.33\textwidth}
\begin{lstgql}
  schema { query: Query }
  
  type Query { topic(name: String): Topic }
  
  type Topic {
    relatedTopics(first: Int): [Topic]
    name: String
    stargazers(after: String, last: Int):
      StargazerConnection }
      
  type StargazerConnection {
    totalCount: Int
    edges: [StargazerEdge]
    nodes: [User] }
    
  type StargazerEdge {
    node: User
    cursor: String }
    
  type User { name: String }
\end{lstgql}
\end{minipage}
\vline
\begin{minipage}[t]{.32\textwidth}
\begin{lstgql}
  query {
    topic(name: "graphql") {
      relatedTopics(first: 2) {
        name
      }
      ...on Starrable {
        stargazers(last: 2, after: "Y3...") {
          totalCount
          edges { # Connections pattern
            node { name }
            cursor
          }
  } } } }
  \end{lstgql}
\end{minipage}
\vline
\begin{minipage}[t]{.25\textwidth}
\begin{lstgql}
  { "data": {
      "topic": {
        "relatedTopics": [
          {"name": "api"},
          {"name": "rest"}
        ],
        "stargazers": {
          "totalCount": 1252,
          "edges": [
            {"node": {"name": "XXX"}, 
             "cursor": "Y3V..."},
            {"node": {"name": "XXX"}, 
             "cursor": "Y3V..."}
          ]
  } } } }
\end{lstgql}
\end{minipage}
  \caption{
    A GraphQL schema (left) with a sample query that uses the connections pattern (center) and response (right).
  }
  \label{fig:qr-example}
  \end{figure*}
}

\section{
  A Novel GraphQL Formalization
}
\label{sec:semantics}

In this section we introduce GraphQL schemas and queries.
Then, we give a novel formalization of the semantics of query execution based on the GraphQL specification~\cite{GraphQLSpec} and reference implementation~\cite{GraphQL-js}.
Compared to Hartig and P\'erez~\cite{Hartig:2018}, our semantics is
   more compact,
   closer to the concrete GraphQL syntax,
   and
   includes the context object for GraphQL-convention-aware static analysis.


\subsection{GraphQL Schemas and Queries}
\label{sec:background_schema}

For a visual introduction to GraphQL queries, see~\cref{fig:qr-example}.
On the left is an excerpt of GitHub's API's schema.
In the center is a sample query requesting
  a \gql{topic} named \gql{"graphql"},
  the \gql{name}s of two \gql{relatedTopics},
  the \gql{totalCount} of stargazers,
  and the \gql{name}s of two stargazers.
The right side of the figure shows the server's response. 

A GraphQL \emph{schema} defines the data \emph{types} that clients can query, as well as possible \emph{operations} on that data.
Types can be scalars (\eg \gql{Int}, \gql{String}), enumerations, object types defined by the provider (\eg \gql{Topic}), or lists (\eg \gql{[Topic]}).
In addition, there are also input types, used to define object types for arguments.
Each field of an object type is characterized by a name (\eg \gql{relatedTopics}) and arguments used to constrain the matching data (\eg \gql{first}).
All schemas define a \gql{Query} operation type, which contains fields that form the top-level entry points for queries~\cite{GraphQLDocs-QueryAndMutationTypes}. Schemas may also define a \gql{Mutation} operation type, which contains fields that allow queries to create, update, or delete data, or a \gql{Subscription} operation type, which provides event-based functionality.

The syntax of a GraphQL \emph{query} is as follows:\footnote{We omit some ``syntactic sugar'' of GraphQL constructs. They can be expressed as combinations of our kernel elements.}
\begin{align}
   q \quad ::= \quad & \gqmatch{label}{field}{args}  \tag{basic query} \\
    \mid \quad & \gqon{type}{q} \tag{filter} \\
    \mid \quad & \gqconj{q}{q} \tag{concatenation} \\
    \mid \quad & \gqnest{\gqmatch{label}{field}{args}}{q} \tag{nesting}
\end{align}

A \textbf{basic query}, \gqmatch{label}{field}{args}, requests a specific \textit{field} of an object with a list of named arguments $\textit{args} = a_1 \texttt{:} v_1 \texttt{,} \dots \texttt{,} a_n \texttt{:} v_n$.
For example \gql{topic(name: "graphql")} in~\cref{fig:qr-example} queries the \gql{topic} field with the argument \gql{name: "graphql"}.
The $\textit{label}$ renames the result of the query with an arbitrary name.
In the GraphQL syntax, $\textit{label}$ can be omitted if it matches the $\textit{field}$, and the list of arguments can be omitted if empty.
A simple $\textit{field}$ is thus a valid query.

Inline fragments $\gqon{type}{q}$ \textbf{filter} a query~$q$ on a type condition, only executing query~$q$ for objects of the correct $\textit{type}$, e.g., \gql{...on Starrable} in \cref{fig:qr-example}.
A query can also \textbf{concatenate} fields~$\gqconj{q_1}{q_2}$,
or request a sub-field of an object via \textbf{nesting} with the form $\gqnest{\gqmatch{label}{field}{args}}{q}$.
Before execution, GraphQL servers validate incoming queries against their schema.

\subsection{Query Execution}
\label{sec:background_execution}

To support a schema, a GraphQL server must implement \textit{resolver functions}.
Each field of each type in the schema corresponds to a resolver in the GraphQL backend. 
The GraphQL runtime invokes the resolvers for each field in the query, and returns a data object that mirrors the shape of the query. 

The evaluation of a query can be thought of as applying successive filters to a \emph{virtual data object} that initially corresponds to the complete data graph.
These filters follow the structure of the query and return only the relevant fields.
For example, when executing the query in \cref{fig:qr-example}, given the \gql{Topic} whose name is \gql{"graphql"}, the resolver for field \gql{relatedTopics} returns a list of \gql{Topic}s, and for each of these \gql{Topic}s the resolver for \gql{name} returns a \gql{String}.

The indirection of resolver functions makes the semantics of GraphQL agnostic to the storage of the data.
The data object is an access point, populated \eg from a database or external service(s) that a resolver contacts.
A resolver must return a value of the appropriate type, but the origin of that value is up to the implementation.

\paragraph{Semantics.}
Formally, \cref{fig:semantics} defines the semantics of the kernel language as an inductive function over the query structure.
The formula $\gsem{q}{o}{ctx} = o'$ means that the evaluation of query~$q$ on data object~$o$ with context~$\textit{ctx}$ returns an object~$o'$.
The context tracks information for use deeper in a nested query.
In our simplified semantics we track the type, field, and arguments of the parent.

{
\begin{figure}
  \centering
  \begin{small}
  $$
  \begin{array}{@{}l@{}}
  \gsem{\gqmatch{label}{field}{args}}{o}{ctx} = \jsobject{\jsfield{label}{\jsresolve{o}{field}{args}{ctx}}}
  \\\\
  \gsem{\gqon{type}{q}}{o}{ctx} =
  \left \{ \begin{array}{ll}
      \gsem{q}{o}{ctx} &\textrm{if } \jstypeof{o} = \textit{type}\\
      \jsobject{} &\textrm{otherwise}
  \end{array}\right.
  \\\\
  \gsem{\gqconj{q_1}{q_2}}{o}{ctx} = \jsmerge{\gsem{q_1}{o}{ctx}}{\gsem{q_2}{o}{ctx}}
  \\\\
  \gsem{\gqnest{\gqmatch{label}{field}{args}}{q}}{o}{ctx} = \\
  \quad
  \left \{ \begin{array}{ll}
    \jsobject{\jsfield{label}{\jslist{\gsem{q}{o_1}{ctx'}}{\gsem{q}{o_n}{ctx'}}}} & \textrm{if } o' = \jslist{o_1}{o_n} \\
    \jsobject{\jsfield{label}{\gsem{q}{o'}{ctx'}}} & \textrm{otherwise}
  \end{array}\right.\\[1.1em]
  \textrm{\ \ \ where } o' = \jsresolve{o}{field}{args}{ctx}\\
  \textrm{\ \ \ and } \textit{ctx}' = \jsobject{
    \jsfield{\texttt{type}}{\jstypeof{o}}\texttt{,}\;
    \jsfield{\texttt{field}}{\textit{field}}\texttt{,}\;
    \jsfield{\texttt{args}}{\textit{args}}
    }
  \end{array}
  $$
  \end{small}
  \vspace{-1em}
\caption{Semantics of GraphQL.}
\label{fig:semantics}
\end{figure}
}

Querying a single field $\gsem{\gqmatch{label}{field}{args}}{o}{ctx}$ calls a resolver function $\jsresolve{o}{field}{args}{ctx}$ which returns the corresponding~$\textit{field}$ in object~$o$. The response object contains a single field $\textit{label}$ populated with this value.
The interpetation the arguments $\textit{args}$ in the resolver is not part of the semantics; it is left to the service developers.

A fragment $\gsem{\gqon{type}{q}}{o}{ctx}$ only evaluates the sub-query~$q$ on objects of the correct $\textit{type}$ (\jstypeof{o} returns the type of its operand).
In the example of \cref{fig:qr-example}, the field \gql{stargazers} is only present in the response if the topic is a \gql{Starrable} type.

Querying multiple fields $\gsem{\gqconj{q_1}{q_2}}{o}{ctx}$ merges the returned objects, collapsing repeated fields in the response into one.\footnote{\jsmerge{o_1}{o_2} recursively merges the fields of~$o_1$ and~$o_2$.}

A nested query $\gsem{\gqnest{\gqmatch{label}{field}{args}}{q}}{o}{ctx}$ is evaluated in two steps.
First, $\jsresolve{o}{field}{args}{ctx}$ returns an object~$o'$.
The second step depends on the type of~$o'$.
If~$o'$ is a list $\jslist{o_1}{o_n}$, the returned object contains a field~$\textit{label}$ whose value is the list obtained by applying the sub-query~$q$ to the all the elements~$o_1, \dots, o_n$ with a new context~$\textit{ctx}'$ containing the type, field, and arguments list of the parent.
Otherwise, the returned object contains a field~$\textit{label}$ whose value is the object returned by applying the sub-query~$q$ on~$o'$ in the new context $\textit{ctx}'$.
By convention the top-level field of the response, which corresponds to the \gql{query} resolver, has an implicit label~\gql{"data"} (see \eg the response in \Cref{fig:qr-example}).


\section{Query Complexity}
\label{sec:background_complexity}

A GraphQL query describes the structure of the response data, and also dictates the resolver functions that must be invoked to satisfy it (which resolvers, in what order, and how many times).
We propose two complexity metrics intended to measure costs from the perspectives of a GraphQL service provider and a client:

\begin{description}
  \item[\textbf{Resolve complexity}] reflects the server's query execution cost.
  \item[\textbf{Type complexity}] reflects the size of the data retrieved by a query.
\end{description}

GraphQL service providers will benefit from either measure, \eg leveraging them to inform load balancing, threat-prevention, resolver resource allocation, or request pricing based on the execution cost or response size.
GraphQL clients will benefit from understanding the type complexity of a query, which may affect their contracts with GraphQL services and network providers, or their caching policies.


Complexity metrics can be computed on either a query or its response.
For a \textbf{query}, in~\cref{sec:analysis-complexity} we propose static analyses to estimate resolve and type complexities \emph{before} its execution given minimal assumptions on the GraphQL server.
For a \textbf{response}, resolve and type complexity are determined similarly but in terms of the fields and data in the response object. 

The intuition behind our analysis is straightforward.
A GraphQL query describes the size and shape of the response.
With an appropriate formalization of GraphQL semantics, an upper bound on resolve complexity and type complexity can be calculated using weighted recursive sums.
But unless it accounts for common GraphQL design practices, the resulting bound may mis-estimate complexities.
In~\cref{sec:evaluation_comparison} we show this problem in existing approaches.

In the remainder of this section, we describe two commonly-used GraphQL pagination mechanisms.
If a GraphQL schema and query uses these mechanisms, either explicitly (\cref{sec:pagination}) or implicitly (\cref{sec:analysis_configuration}), we can obtain a tighter and thus more useful complexity bound.
Research reported that both of these conventions are widely used in real-world GraphQL schemas~\cite{Wittern:2019}, so supporting them is also important for practical purposes.

\subsection{GraphQL Pagination Conventions}
\label{sec:pagination}

At the scale of commercial GraphQL APIs, queries for fields that return lists of objects may have high complexity --- \eg consider the (very large) cross product of all GitHub repositories and users. 
The official GraphQL documentation recommends that schema developers bound response sizes through pagination, using \textit{slicing} or the \textit{connections pattern}~\cite{GraphQLDocs-Pagination}.
GraphQL does not specify semantics for such arguments, so we describe the common convention followed by commercial~\cite{GitHubGraphQL:2019,ShopifyGraphQL:2019,YelpGraphQL:2017} and open-source~\cite{Wittern:2019} GraphQL APIs.

Resolvers can return lists of objects which can result in arbitrarily large responses -- bounded only by the size of the underlying data.
\textbf{Slicing} is a solution that uses \emph{limit arguments} to bound the size of the returned lists (\eg \gql!relatedTopics(first: 2)! in \cref{fig:qr-example}).

The \textbf{connections pattern} introduces a layer of indirection for more flexible pagination, using virtual \gql{Edge} and \gql{Connection} types~\cite{GraphQLDocs-Pagination, Relay-Pagination}.
For example, in~\cref{fig:qr-example}-left the field \gql{stargazers} returns a single \gql{StargazerConnection}, allowing access to the \gql{totalCount} of stargazers and the \gql{edges} field, returning a list of \gql{StargazerEdge}s.
This pattern requires limit arguments to target children of a returned object (\eg \gql!stargazers(last: 2)! in~\cref{fig:qr-example}-middle applies to the field \gql{edges}).

The size of a list returned by a resolver can thus depend on the current arguments and the arguments of the parent stored in the context. 
Ensuring that limit arguments actually bound the size of the returned list is the responsibility of the server developers:

\begin{assumption}\label{hyp:resolve-limit}
  If the arguments list~($\mathit{args}$) or the context~($\mathit{ctx}$) contains a limit argument ($\mathit{arg} \texttt{:} \mathit{val}$) the list returned by the resolver cannot be longer than the value of the argument~($\mathit{val}$), that is:
  \begin{center}
  \begin{small}
  $\jslength{\jsresolve{o}{field}{args}{ctx}} \leq \mathit{val}.$
  \end{small}
\end{center}
\end{assumption}

\noindent
If this assumption fails, it likely implies a backend error.


\paragraph{Pagination, not panacea.}
While slicing and the connections pattern help to constrain the response size of a query and the number of resolver functions that its execution invokes, these patterns cannot prevent clients from formulating complex queries that may exceed user rate limits or overload backend systems.
Our pagination-aware complexity analyses can statically identify such queries.

\subsection{Configuration for Pagination Conventions}\label{sec:analysis_configuration}

As we discuss in~\cref{sec:evaluation}, ignoring slicing arguments or mis-handling the connections pattern can lead to under- or over-estimation of a query's cost.
Understanding pagination semantics is thus essential for accurate static analysis of query complexity.
Since GraphQL pagination is a convention rather than a specification, we therefore propose to complement GraphQL schemas with a configuration that captures common pagination semantics.
To unify this configuration with our definitions of resolve and type complexity, we also include weights representing resolver and type costs.
%
%
Here is a sample configuration for the schema from~\cref{fig:qr-example}:

\noindent
\begin{center}
\begin{minipage}[t]{.24\textwidth}
\begin{lstgql}
resolvers:
  "Topic.relatedTopics":
    limitArguments: [first]
    defaultLimit: 10
    resolverWeight: 1
  "Topic.stargazers":
    limitArguments: [first, last]
    limitedFields: [edges, nodes]
    defaultLimit: 10
    resolverWeight: 1
\end{lstgql}
\end{minipage}
\vline
\hspace{.02\textwidth}
\begin{minipage}[t]{.18\textwidth}
\begin{lstgql}
types:
  Topic:
    typeWeight: 1
  Stargazer:
    typeWeight: 1





\end{lstgql}
\end{minipage}
\end{center}

\medskip

This configuration specifies pagination behavior for slicing and the connections pattern.
In this configuration, resolvers are identified by a string \jsstring{type.field} (\eg \gql{"Topic.relatedTopics"}).
Their limit arguments are defined with the field \gql{limitArguments}.
For slicing, the limit argument applies directly to the returned list (see \gql{"Topic.relatedTopics"}).
For the connections pattern, the limit argument(s) apply to children of the returned object (see \gql{limitedFields} for \gql{"Topic.stargazers"}).
The \gql{defaultLimit} field indicates the size of the returned list if the resolver is called without limit arguments.
We must make a second assumption (using JavaScript dot and bracket notation to access the fields of an object):

\begin{assumption}\label{hyp:resolve-default}
  If a resolver is called without limit arguments, the returned list is no longer than the configuration $\textrm{c}$'s default limit.
  \begin{center}
  \begin{small}
  $
  \jslength{\jsresolve{o}{field}{args}{ctx}} \leq 
  \jsaccess{\jsget{\jsaccess{c}{\texttt{resolvers}}}{\jsstring{type.field}}}{\texttt{defaultLimit}}
  $
  \end{small}
\end{center}
\end{assumption}

\noindent
In the following, $\jslimit{c}{type}{field}{args}{ctx}$ returns the maximum value of
  the limit arguments for the resolver~$\jsstring{type.field}$ if such arguments are present in the arguments list~$\textit{args}$ or the context~$\textit{ctx}$,
  and the default limit otherwise.
If a resolver returns unbounded lists, the default limit can be set to~$\infty$, but we urge service providers to always bound lists for their own protection.

From Assumptions~\ref{hyp:resolve-limit} and~\ref{hyp:resolve-default}, we have:

\begin{property}\label{prop:resolve-limit}
  Given a configuration~$\textit{c}$ and a data object~$o$ of type~$t$, if the context~$\textit{ctx}$ contains the information on the parent of~$o$, we have:
  \begin{center}
  \begin{small}
    $
    \jslength{\jsresolve{o}{field}{args}{ctx}} \leq \jslimit{c}{t}{field}{args}{ctx}.
    $
  \end{small}
  \end{center}

\end{property}

\MyParagraph{Concise configuration.}
Researchers have reported that many \\ GraphQL schemas follow consistent naming conventions~\cite{Wittern:2019}, so we believe that regular expressions and wildcards are a natural way to make a configuration more concise.
For example, the expression \gql{"/.*Edge$/.nodes"} can be used for configuring resolvers associated with \gql{nodes} fields within all types whose names end in \gql{Edge}.
Or, the expression \gql{"User.*"} can be used for configuring resolvers for all fields within type~\gql{User}.

\section{GraphQL Query Cost Analysis}
\label{sec:query_analysis}

In this section we formalize two query analyses to estimate the type and resolve complexities of a query.
The analyses are defined as inductive functions over the structure of the query, mirroring the formalization of the semantics presented in \cref{sec:semantics}.
We highlight applications of these complexities in~\cref{sec:evaluation_comparison}.


Like other static query cost analyses (\cref{sec:bm-existingAnalyses}),
  our analysis returns upper bounds for the actual response costs.
For example, if a query asks for 10 \gql{Topic}s, our analysis will return 10 as a tight upper bound on the response size.
If there are only 3 \gql{Topic}s in the data graph, our analysis will over-estimate the actual query cost.

\subsection{Resolve and Type Complexity Analyses}
\label{sec:static_analysis}

As mentioned in \cref{sec:background_complexity}, we propose two complexity metrics: resolve complexity and type complexity.
In this section, we formalize the resolve complexity analysis and then explain how to adapt the approach to compute the type complexity.
We will work several complexity examples, relying on the \textit{$W_{1,0}$ configuration}:
a resolver weight of 1 for fields returning object and list of object and 0 for all other fields and the top-level operation resolvers (i.e. \gql{query}, \gql{mutation}, \gql{subscription}), and a type weight of 1 for all objects (including those returned in lists) and 0 for all other types and the top-level operation types (i.e. \gql{Query}, \gql{Mutation}, \gql{Subscription}).

\MyParagraph{Resolve Complexity.}
Resolve complexity is a measure of the execution costs of a query.
Given a configuration~$\textit{c}$, the resolve complexity of a response~$r$:~$\jsrsize{r}{c}$ is the sum of the weights of the resolvers that are called to compute the response.
Each resolver call populates a field in the response.
The resolve complexity of a response is thus the sum of the weights of each field.
With the $W_{1,0}$ configuration, the resolve complexity of the response in \cref{fig:qr-example}-right is~$6$: 1~\gql{topic}, 1~\gql{relatedTopics}, 1~\gql{stargazers}, 1~\gql{edges}, 2~\gql{nodes}.

The query analysis presented in \cref{fig:complexity} computes an estimate \code{qrcx} of the resolve complexity of the response.
Each call to a resolver is abstracted by the corresponding weight, and the sizes of lists are bounded using the \texttt{limit} function.
\texttt{limit} uses the limit argument if present; otherwise, it will use the default limit in the configuration.

The analysis is defined by induction over the structure of the query and mirrors the semantics presented in \cref{fig:semantics} with one major difference: the complexity analysis operates solely on the types defined in the schema starting from the top-level \gql{Query} type.
The formula $\grsize{c}{q}{t}{ctx} = x$ means that given a configuration~$c$, the estimated complexity of a query~$q$ on a type~$t$ with a context~$\textit{ctx}$ is $x \in \mathbb{N} \cup \{ \infty \}$.
For example, using the $W_{1,0}$ configuration, the resolve complexity \emph{of the query} in \cref{fig:qr-example}-middle is also~$6$.

\begin{figure}
  \begin{small}
  $$
  \begin{array}{@{}l}
  \grsize{c}{\gqmatch{label}{field}{args}}{t}{ctx} =
  \\
  \quad
  \jsaccess{\jsget{\jsaccess{c}{\texttt{resolvers}}}{\jsstring{t.field}}}{\texttt{resolverWeight}}
  
  \\\\
  \grsize{c}{\gqon{type}{q}}{t}{ctx} =
  \left \{ \begin{array}{ll}
      \grsize{c}{q}{t}{ctx} &\textrm{if } t = \textit{type}\\
      0  &\textrm{otherwise}
  \end{array}\right.
  \\\\
  \grsize{c}{\gqconj{q_1}{q_2}}{t}{ctx} = \grsize{c}{q_1}{t}{ctx} + \grsize{c}{q_2}{t}{ctx}
  \\\\
  \grsize{c}{\gqnest{\gqmatch{label}{field}{args}}{q}}{t}{ctx} =\\
  \quad
  \left \{ \begin{array}{ll}
    w + l \times \grsize{c}{q}{t'}{ctx'} & \textrm{if } \jsget{t}{\textit{field}} = \jsisarray{t'}\\
    w + \grsize{c}{q}{\jsaccess{t}{\textit{field}}}{ctx'} & \textrm{otherwise}
  \end{array}\right.\\[1.1em]
  \textrm{where } w = \jsaccess{\jsget{\jsaccess{c}{\texttt{resolvers}}}{\jsstring{t.field}}}{\texttt{resolverWeight}}\\
  \textrm{and } l = \jslimit{c}{t}{field}{args}{ctx}\\
  \textrm{and } \textit{ctx'} = \jsobject{
    \jsfield{\texttt{type}}{t}\texttt{,}\;
    \jsfield{\texttt{field}}{\textit{field}}\texttt{,}\;
    \jsfield{\texttt{args}}{\textit{args}}
    }
  \end{array}
  $$
  \end{small}
  \vspace{-1em}
\caption{Resolve complexity analysis.
The analysis operates on the types defined in the schema starting from \gql{Query}.}
\label{fig:complexity}
\end{figure}

\begin{restatable}{theorem}{saresolve}
\label{thm:resolve-complexity}
  Given a configuration~$\textit{c}$,
  the analysis of \cref{fig:complexity} always returns an upper-bound of the resolve complexity of the response.
  For any query~$q$ over a data object~$o$, and using \texttt{\{\}} to denote the empty initial context:
  \begin{center}
  \begin{small}
  $
  \grsize{c}{q}{\texttt{Query}}{\jsobject{}} \geq \jsrsize{\gsem{q}{o}{\jsobject{}}}{c}
  $
  \end{small}
  \end{center}
\end{restatable}

\begin{proof}
  The theorem is proved by induction on the structure of the query.
  The over-approximation has two causes.
  First, \texttt{limit} returns an upper bound on the size of the returned list (Property~\ref{prop:resolve-limit}).
  If the underlying data is sparse, the list will not be full; the maximum complexity will not be reached.
  Second, we express the complexity of $\jsmerge{o_1}{o_2}$ as the sum of the complexity of the two objects.
  This is accurate if~$o_1$ and~$o_2$ share no properties, but if properties overlap then the merge will remove redundant fields (cf. \cref{sec:background_execution}).
\end{proof}

\MyParagraph{Type Complexity.}
Type complexity is a measure of the size of the response object.
Given a configuration~$c$, the type complexity of a response object~$r$: $\jstsize{r}{c}$ is the sum of the weights of the types of all objects in the response.
Using the $W_{1,0}$ configuration, the type complexity of the response in \cref{fig:qr-example}-right is~$8$: 1~\gql{Topic}, 2~(related) \gql{Topic}s, 1~\gql{StargazerConnection}, 2~\gql{StargazerEdge}s, 2~\gql{User}s.

Similar to our resolve complexity analysis \code{qrcx} (\cref{fig:complexity}),
our type complexity analysis bounds the response's type complexity without performing query execution.
We call the estimated query type complexity \code{qtcx}.
To compute \code{qtcx}, we tweak the first and final rules from the \code{qrcx} analysis:

\begin{enumerate}
\item The call to a resolver is abstracted by the weight of the returned type~$t' = \jsget{t}{\textit{field}}$.
  $$
  \begin{small}
  \begin{array}{@{}l}
  \gtsize{c}{\gqmatch{label}{field}{args}}{t}{ctx} =
  \jsaccess{\jsget{\jsaccess{c}{\texttt{types}}}{t'}}{\texttt{typeWeight}}. 
  \end{array}
  \end{small}
  $$

\item When a nested query returns a list ($\jsget{t}{\textit{field}} = [t']$), the type complexity must reflect the cost of instantiating every element. Every element thus adds the weight of the returned type~$t'$ to the complexity.
$$
\begin{small}
\begin{array}{@{}l}
\gtsize{c}{\gqnest{\gqmatch{label}{field}{args}}{q}}{t}{ctx} =
l \times (w + \gtsize{c}{q}{t'}{ctx'})\\[0.2em]
\textrm{where } w = \jsaccess{\jsget{\jsaccess{c}{\texttt{types}}}{t'}}{\texttt{typeWeight}}.
\end{array}
\end{small}
$$
\end{enumerate}

\noindent
With the $W_{1,0}$ configuration, the type complexity \emph{of the query} in \cref{fig:qr-example} is also~$8$.

\begin{theorem}\label{thm:type-complexity}
  Given a configuration~$\textit{c}$,
  the type complexity analysis always returns an upper-bound of the type complexity of the response.
  For any query~$q$ over a data object~$o$, and using \texttt{\{\}} to denote the empty initial context:
  \begin{center}
  \begin{small}
  $
  \gtsize{c}{q}{\texttt{Query}}{\jsobject{}} \geq \jstsize{\gsem{q}{o}{\jsobject{}}}{c}
  $
  \end{small}
  \end{center}
\end{theorem}

The proof is similar to the proof of Theorem~\ref{thm:resolve-complexity}.

\subsection{Time/Space Complexity of the Analyses} \label{sec:analysis-complexity}
The type and resolve complexity analyses are computed in one traversal of the query.
The time complexity of both analyses is thus~$O(n)$, where~$n$ is the size of the query as measured by the number of derivations required to generate it from the grammar of GraphQL.
Both analyses need to track only the parent of each sub-query during the traversal. This implies that the space required to execute the analyses depends on the maximum nesting of the query which is at worst~$n$.
The space complexity is thus in~$O(n)$.

We emphasize that these are static analyses -- they do not need to communicate with backends.

\subsection{Mutations and Subscriptions}
So far we have considered only GraphQL queries.
GraphQL also supports \emph{mutations} to modify the exposed data and \emph{subscriptions} for event-based functionality~\cite{GraphQLSpec}.
Our resolve complexity approach also applies to mutations, reflecting the execution cost of the mutation.
API providers can use resolve weights to reflect costly mutations in resolve complexity calculations.
However, our type complexity approach only estimates the size of the returned object, ignoring the amount of data modified along the way.
Assuming the user provides the new data, computing the type complexity of arguments passed to mutation resolvers may give a reasonable approximation.
We leave this for future work.
Our analysis can also produce resolve and type complexities for subscription queries. Policies around subscriptions may differ, though. For rate-limiting, for example, the API provider could reduce the remaining rates based on complexities when a subscription is started, and replenish them once the client unsubscribes.



\section{Evaluation}
\label{sec:evaluation}
We have presented our query analysis and proved its accuracy.
In our evaluation we consider five questions:

\begin{enumerate}[leftmargin=0.75cm]
  \item[\textbf{RQ1}:] Can the analysis be \textit{applied to real-world GraphQL APIs}, especially considering the required configuration?
  \item[\textbf{RQ2}:] Does the analysis produce \textit{cost upper-bounds} for queries to such APIs?
  \item[\textbf{RQ3}:] Are these bounds \textit{useful}, i.e. close enough to the actual costs of the queries for providers to respond based on the estimates?
  \item[\textbf{RQ4}:] Is our analysis \textit{cheap enough} for use in API management?
  \item[\textbf{RQ5}:] How does our approach \textit{compare} to other solutions? 
\end{enumerate}

{
\setlength{\belowcaptionskip}{-0em} 
\begin{table}
  \centering
  \caption{Characteristics of the evaluated APIs.}
  \vspace{-0.5em}
  \begin{tabular}{@{}p{0.25\textwidth} r @{\hspace{3em}} r@{}} 
    \textsc{Schema}                   & \textsc{GitHub} & \textsc{Yelp} \\
    \toprule
    Number of object types            & 245             & 25 \\
    Total fields on all object types  & 1,569           & 121 \\
    Lines of Code (LoC)               & 22,071          & 760 \\
    Pagination w. slicing arguments   & Yes             & Yes \\
    Pagination w. connections pattern & Yes             & Yes \\
\\
    \textsc{Configuration}            & \textsc{GitHub} & \textsc{Yelp} \\
    \toprule
    Number of default limits          & 21 & 13  \\
    LoC (\% of schema LoC )           & 50 (0.2\%) & 39 (5.1\%)
  \end{tabular}
  \label{table:evaluation}
  \end{table}
}


Our analysis depends on a novel dataset (\cref{sec:evaluation_query-generation}).
We created the first GraphQL query-response corpus for two reputable, publicly accessible GraphQL APIs: GitHub~\cite{GitHubGraphQL:2019} and Yelp~\cite{YelpGraphQL:2017}.
\cref{table:evaluation} summarizes these APIs using the metrics from~\cite{Wittern:2019}.

To answer RQ1, we discuss configuring our analysis for these APIs (\cref{sec:evaluation_configuration}).
To answer RQ2-RQ4, we analyzed the predicted and actual complexities in our query-response corpus (\cref{sec:evaluation_measurements}).
For RQ5, we compare our findings
  experimentally to open-source static analyses,
  and
  qualitatively to closed-source commercial approaches (\cref{sec:evaluation_comparison}).

\subsection{A GraphQL Query-Response Corpus}
\label{sec:evaluation_query-generation}
Answering our research questions requires a multi-API query-response corpus.
None existed so we created one.
We automatically generated queries for GitHub and Yelp,
  and collected $5,000$ unique query-response pairs each from April 2019 to March 2020. 

\MyParagraph{General approach.}
We developed and open-sourced a GraphQL query generator~\cite{artifact}.
Its input is a GraphQL schema and user-supplied parameters.
Following standard grammar-based input generation techniques~\cite{godefroid2008grammar}, it generates queries recursively and randomly by building an abstract syntax tree and rendering it as a GraphQL query string.
The \textit{depth} probability dictates the likelihood of adding nested fields.
The \textit{breadth} probability controls the likelihood of adding additional adjacent fields into each query level.

\MyParagraph{Providing Argument Values.}
The AST generated by this approach includes \textit{argument variables} that must be concretized.
We populate these values using a \emph{provider map} that maps a GraphQL argument to functions that generate valid values.


\MyParagraph{Queries for GitHub and Yelp.}
We tuned the query generation to our experimental design and context.
We used breadth and depth probabilities of $0.5$.
For \textit{numeric arguments}, we used normally distributed integer values with mean and variance of $\langle 5, 1.0 \rangle$ (GitHub) and $\langle 7, 2.0 \rangle$ (Yelp).\footnote{Yelp has a maximum nesting depth. Larger values yielded complex, shallower queries.}
For \textit{enum arguments}, we randomly chose a valid enum value.
For \textit{named arguments} (\eg GitHub projects, Yelp restaurants), we randomly chose a valid entity, \eg the $100$ most-starred GitHub repositories. 
For \textit{context-dependent arguments}, we used commonly-valid values, \eg \code{README.md} as a GitHub filepath.


Ethically, because we are issuing queries against public APIs, we omitted overly-complex queries and queries with side effects.
We issued queries with:
  (1) estimated type or resolve complexity $\leq 1,000$;
  (2) nesting depth $\leq 10$;
  and
  (3) \gql{query}, not \gql{mutation} or \gql{subscription} operations. 

\newcommand{\MedianQueryTypeComplexity}{81.5\xspace}

\newcommand{\MedianResolveTypeComplexity}{29\xspace}

\MyParagraph{Query realism.}
Our query generation algorithm yielded diverse queries for each API.
To assess their realism, we compared them to example queries provided by each API. 
Yelp did not provide any examples, but GitHub's documentation includes 24 example queries~\cite{GitHubGraphQL:Examples:2019}. 
We analyzed them using the \textit{$W_{1,0}$} configuration (cf. \Cref{sec:static_analysis}) to determine the size of a typical query.
Only $14$ could be analyzed using the same GitHub schema version. 
Of these, all but one had a type complexity of $302$ or less.
Therefore, we used a type complexity of $\leq 300$ to categorize \textit{typical queries}.



\subsection{RQ1: Configuration}
\label{sec:evaluation_configuration}
Our aim is a backend-agnostic query cost analysis.
As prior work~\cite{Hartig:2018,Andersson2018GQLResultSizeCalculation} cannot be applied to estimate query cost for arbitrary GraphQL servers, we first assess the feasibility of our approach.

Per~\cref{sec:analysis_configuration}, our analysis requires GraphQL API providers to configure limit arguments, weights, and default limits.
We found it straightforward to configure both APIs.
Our configurations are far smaller than the corresponding schemas (\cref{table:evaluation}).

\MyParagraph{Limit arguments.}
GitHub and Yelp uses consistent limit arguments names (\gql{first} and \gql{last} for GitHub, \gql{limit} for Yelp), which we configured with regular expressions.
GitHub has a few exceptional limit arguments, such as \gql{limit} on the \gql{Gist.files} field, which we configured with both strings and regular expressions.
Following the conventions of the connections pattern~\cite{Relay-Pagination}, we set all \gql{Connection} types to have the limited fields \gql{edges} and \gql{nodes} for GitHub.
GitHub also has a few fields that follow the slicing pattern, so we did not set any limited fields for these.
The Yelp API strictly follows the slicing pattern, so no additional settings were required. 

\MyParagraph{Weights.}
For simplicity, we used a \textit{$W_{1,0}$} configuration (cf. \cref{sec:static_analysis}).
This decision permitted us to compare against the open-source static analysis libraries, but may not reflect the actual costs for GraphQL clients or service providers.
For example, we set the type weights for scalars and enums to~$0$, supposing that these fields are paid for by the resolvers for the containing objects.

\MyParagraph{Default limits.}
We identified default limits for the 21 fields (GitHub) and 13 fields (Yelp).
These fields are unpaginated lists or lists that do not require limit arguments.
We determined these numbers using API documentation and experimental queries.

\subsection{RQ2-RQ4: Complexity Measurements}
\label{sec:evaluation_measurements}

Using this configuration, we calculated type and resolve complexities for each query-response pair in the corpus.
\cref{fig:RQ1PredQuad} summarizes the results for Yelp and GitHub using heat maps.
Each heat map shows the density of predictions for request/response complexity cells.
Cells above the diagonal are queries whose actual complexity we over-estimated, cells below the diagonal are under-estimates.


\subsubsection{RQ2: Upper Bounds}
\label{sec:evaluation_measurements_complexities}

In~\cref{fig:RQ1PredQuad}, all points lie on or above the diagonal.
In terms of our analysis, every query's predicted complexity is an upper bound on its actual complexity.
This observation is experimental evidence for the results theorized in~\cref{sec:static_analysis}.

As an additional note, the striations that can be seen in~
\cref{fig:OUR_complexity_heatmap_resolve_yelp} are the result of fields with large default limits. Queries that utilize these fields will have similar estimated complexities.

\begin{figure}
  \centering
  \subfloat[Yelp type complexities]{
    \label{fig:OUR_complexity_heatmap_type_yelp}
    \includegraphics[width=0.5\columnwidth]{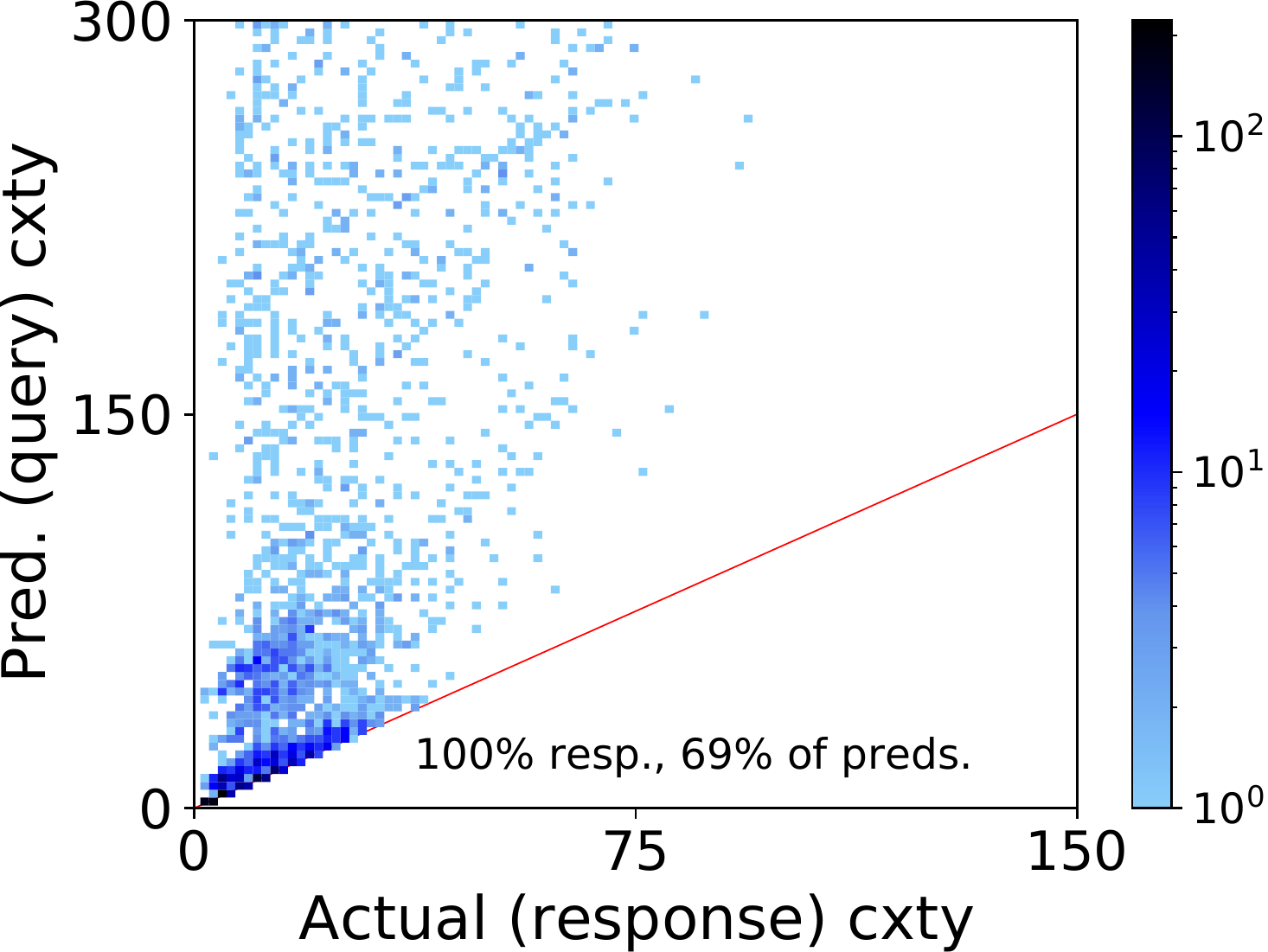}
  }
  \subfloat[Yelp resolve complexities]{
    \label{fig:OUR_complexity_heatmap_resolve_yelp}
    \includegraphics[width=0.5\columnwidth]{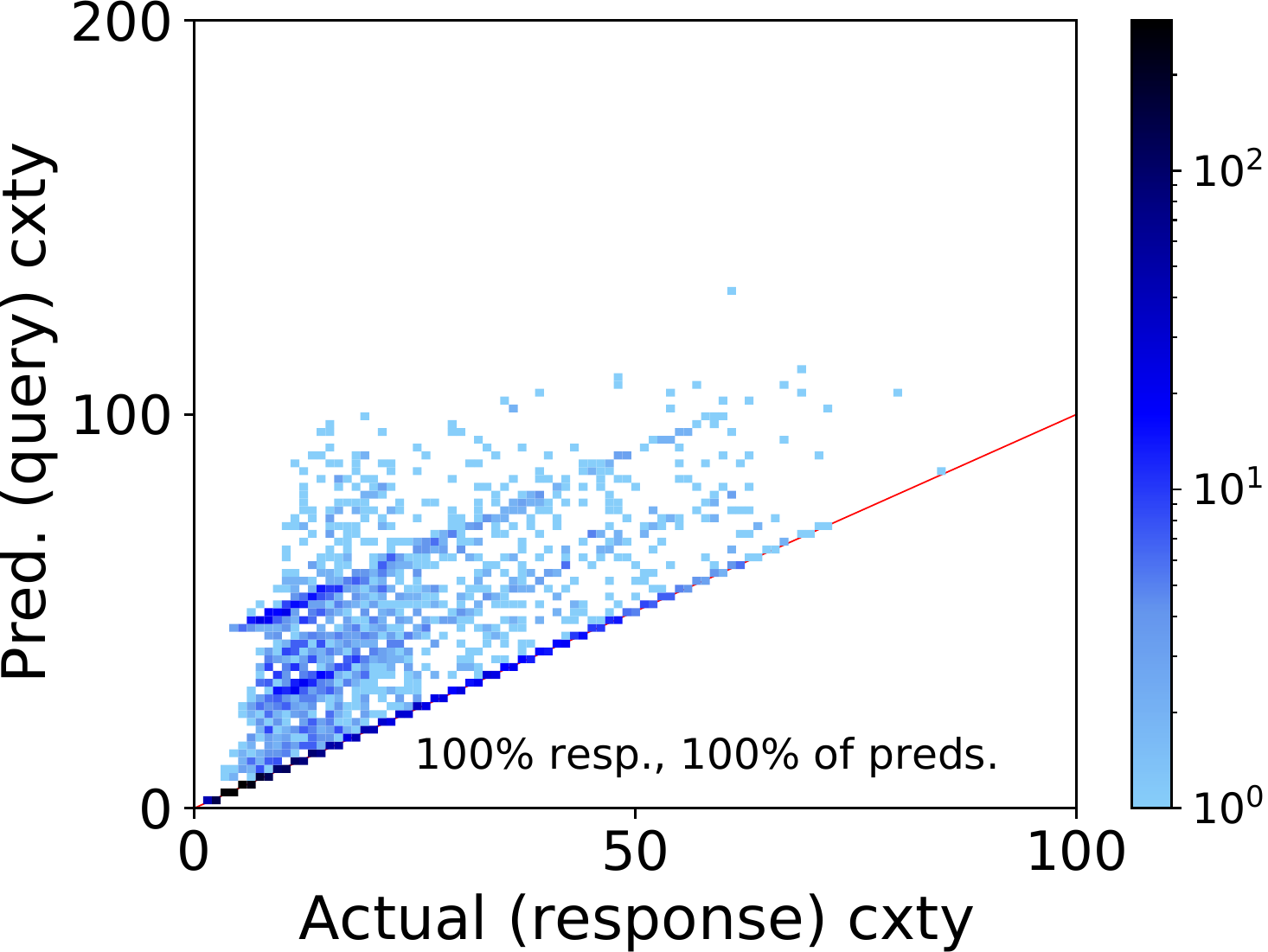}
  }
  \\
  \subfloat[GitHub type complexities]{
    \label{fig:OUR_complexity_heatmap_type_github}
    \includegraphics[width=0.5\columnwidth]{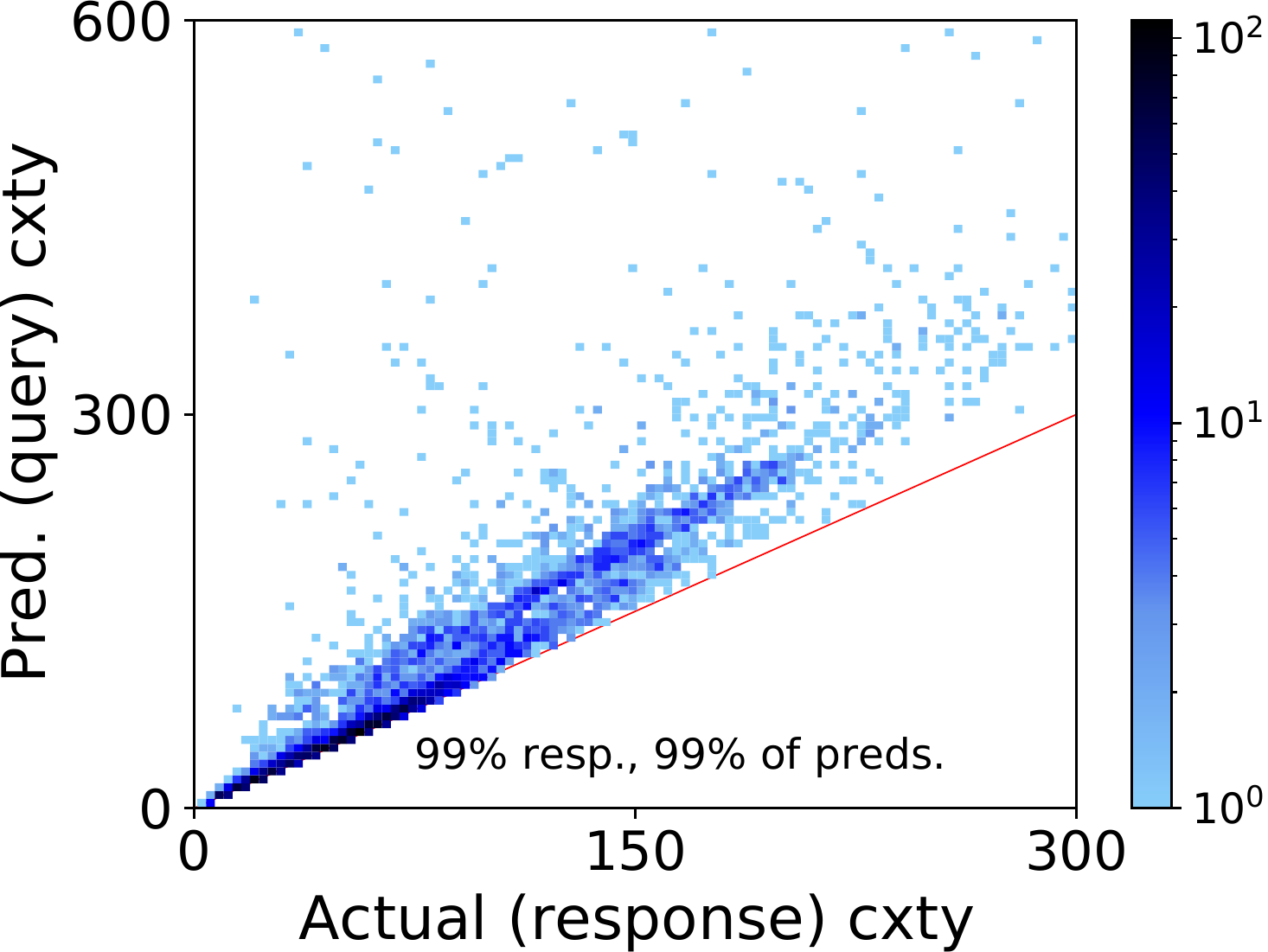}
  }
  \subfloat[GitHub resolve complexities]{
    \label{fig:OUR_complexity_heatmap_resolve_github}
    \includegraphics[width=0.5\columnwidth]{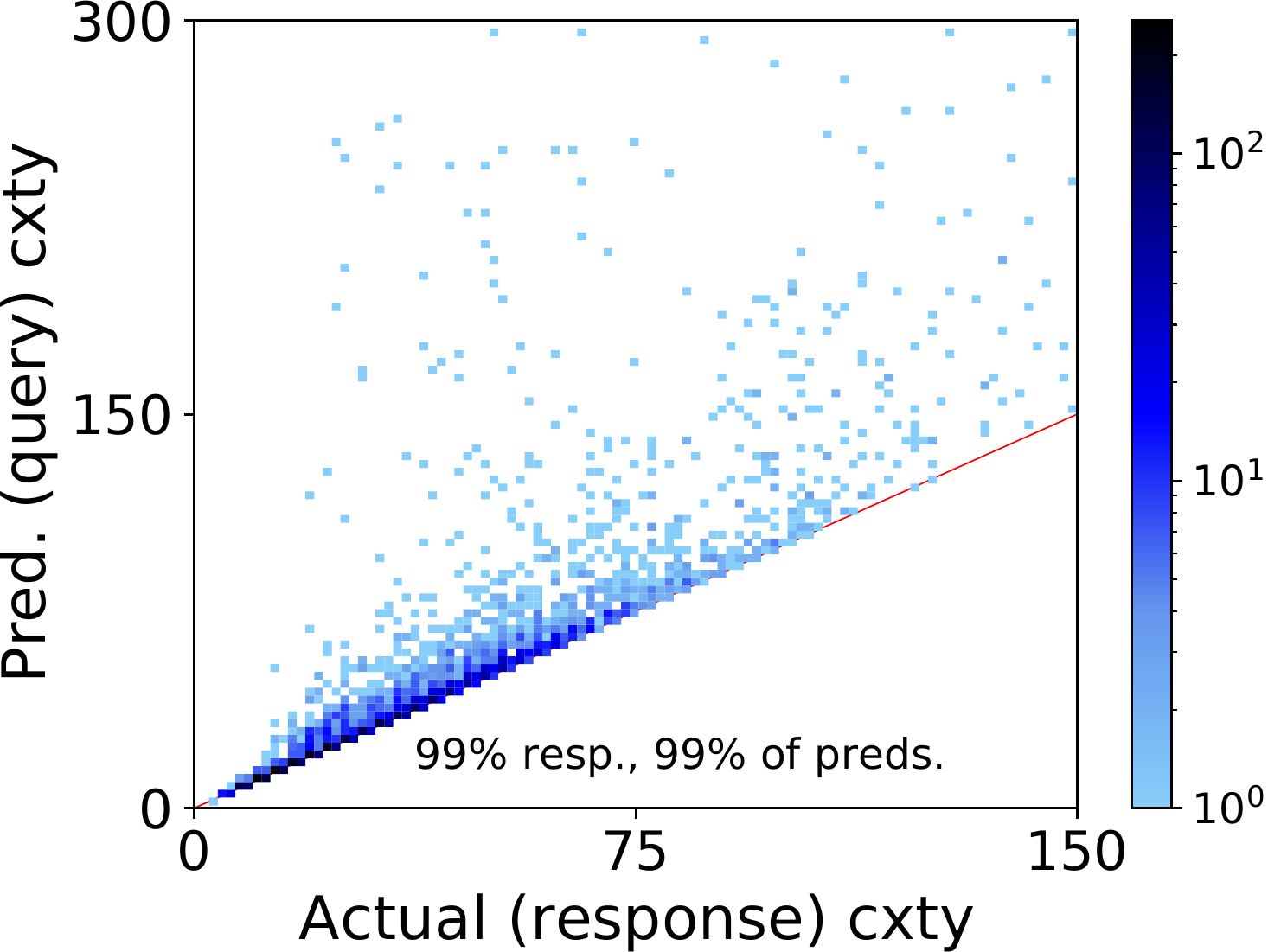}
  }
  \caption{
    Actual (response) complexities and predicted (query) complexities,
    using our analysis on the corpus.
    Each figure has text indicating the percentage of responses that are shown (the remainder exceed the x-axis),
      and the percentage of the corresponding predictions that are shown (the remainder exceed the y-axis).
  }
  \label{fig:RQ1PredQuad}
\end{figure}

\subsubsection{RQ3: Tight Upper Bounds}
\label{sec:evaluation_measurements_tight-bound}

{
\setlength{\belowcaptionskip}{-0em} 
\begin{table}
  \centering
  \caption{Over-estimation of our type and resolve complexities.
   The first table summarizes our approach on all queries.
   The second shows the results for ``typical'' queries, \ie those with type complexities of up to 300 (cf. \cref{sec:evaluation_query-generation}).
  }
  \label{table:overestimation}
  \vspace{-0.5em}
  \begin{tabular}{l c c c c} 
     & \multicolumn{2}{c}{\textsc{Yelp ($5,000$)}} & \multicolumn{2}{c}{\textsc{GitHub ($5,000$)}} \\
    \textsc{All queries}   & Resolve & Type & Resolve & Type \\
    \toprule
    Underestimation         & \textit{None}  & \textit{None}  & \textit{None}  & \textit{None} \\ 
    No overestimation       & 60.1\% & 21.8\% & 54.7\% & 8.4\% \\ 
    Overestimation <25\%    & 63.0\% & 26.1\% & 84.9\% & 60.3\% \\
    Overestimation <50\%    & 66.7\% & 30.1\% & 92.2\% & 84.0\% \\
    \\
     & \multicolumn{2}{c}{\textsc{Yelp ($3,416$)}} & \multicolumn{2}{c}{\textsc{GitHub ($4,703$)}} \\
    \textsc{``Typical'' queries}  & Resolve & Type & Resolve & Type \\
    \toprule
    Underestimation         & \textit{None}  & \textit{None}  & \textit{None}  & \textit{None} \\ 
    No overestimation       & 83.4\% & 31.9\% & 57.9\% & 9.0\% \\ 
    Overestimation <25\%    & 85.3\% & 38.2\% & 88.3\% & 63.8\% \\
    Overestimation <50\%    & 88.1\% & 44.1\% & 95.2\% & 87.3\% \\
  \end{tabular}

\end{table}
}

Answering RQ2, we found our analysis computes upper bounds on the type and resolve complexities of queries.
Other researchers have suggested that these cost bounds may be too far from the actual costs to be actionable~\cite{Hartig:2018}.

Our data show that the bounds are tight enough for practical purposes, and are as tight as possible with a static, data-agnostic approach.
\cref{fig:RQ1PredQuad} indicates that our upper bound is close or equal to the actual type and resolve complexities of many queries --- this can be seen in the high density of queries near the diagonals.

Our bounds are looser for more complex queries.
This follows intuition about the underlying graph:
  larger, more nested queries may not be satisfiable by an API's real data.
Data sparsity leads responses to be less complex than their worst-case potential.
\cref{table:overestimation} quantifies this observation.
It shows the share of queries for which our predictions over-estimate the response complexity by <25\% and <50\%.
Over-estimation is less common for the subset of ``typical'' queries whose estimated type complexity is $\leq 300$.

However, per the proofs in~\cref{sec:analysis-complexity}, our upper bounds are as tight as possible without dynamic response size information.
The over-estimates for larger queries are due to data sparsity, not inaccuracy.
For example, consider this pathological query to GitHub's API:

\begin{lstgql}
  query {
    organization (login: "nodejs") {
      repository (name: "node") { issues (first: 100) { nodes {
        repository { issues (first: 100) { nodes {
          ... }}}}}}}}
\end{lstgql}

\noindent This query cyclically requests the same \gql{repository} and \gql{issues}.
%
With two levels of nesting, the query complexities are $10,203$ (resolve) and $20,202$ (type).
If the API data includes at least 100 \gql{issues}, the response complexities will match the query complexities.

\subsubsection{RQ4: Performance}
\label{sec:evaluation_measurements_performance}
Beyond the functional correctness of our analysis, we assessed its runtime cost to see if it can be incorporated into existing request-response flows, \eg in a GraphQL client or an API gateway.
We measured runtime cost on a 2017 MacBook Pro (8-core Intel i7 processor, 16 GB of memory).

As predicted in~\cref{sec:analysis-complexity}, our analysis runs in linear time as a function of the query and response size.\footnote{We define \emph{query size} as the number of derivations required to generate the query from the grammar presented in \cref{sec:background_schema}. It can be understood as the number of lines in the query if fields, inline fragments, and closing brackets each claim their own line.}
The median processing time was $3.0$~ms for queries, and $1.1$~ms for responses.
Even the most complex inputs were fairly cheap; $95\%$ of the queries could be processed in $<7.3$~ms, and $95\%$ of the responses in $<4$ ms.
The open-source analyses we consider in~\cref{sec:evaluation_comparison} also appear to run in linear time.


\subsection{RQ5: Comparison to Other Static Analyses}
\label{sec:evaluation_comparison}

\newcommand{\BQOne}{BQ1\xspace}
\newcommand{\BQTwo}{BQ2\xspace}

In this section we compare our approach to state-of-the-art static GraphQL analyses (\cref{sec:bm-existingAnalyses}).
To permit a fair comparison across different notions of GraphQL query cost, we tried to answer two practical \emph{bound questions} for the GitHub API with each approach.

\begin{enumerate}[leftmargin=1cm]
  \item[\textbf{BQ1}:] How large might the response be?
  \item[\textbf{BQ2}:] How many resolver functions might be invoked?
\end{enumerate}

\noindent
\BQOne is of interest to clients and service providers, who both pay the cost of handling the response.
Various interpretations of ``large'' are possible, so we operationalized this as the number of distinct objects (non-scalars) in the response.
\BQTwo is of interest to service providers, who pay this cost when generating the response.

We configured and compared our analysis against three open-source analyses experimentally, with results shown in~\cref{fig:OSSComparison_GitHub}.
The static GraphQL analyses performed by corporations are not publicly available, so we discuss them qualitatively instead.

\newcommand{\LibCatalyzerGVC}{libA\xspace}
\newcommand{\LibSlicknodeGQC}{libB\xspace}
\newcommand{\LibPaBruGCA}{libC\xspace}

{
 \setlength{\belowcaptionskip}{-1em} 
\begin{figure*}[!ht]
  \centering
  \subfloat[Our GitHub complexities]{
    \label{fig:OSSComparison_GitHub_Us}
    \includegraphics[width=0.50\columnwidth]{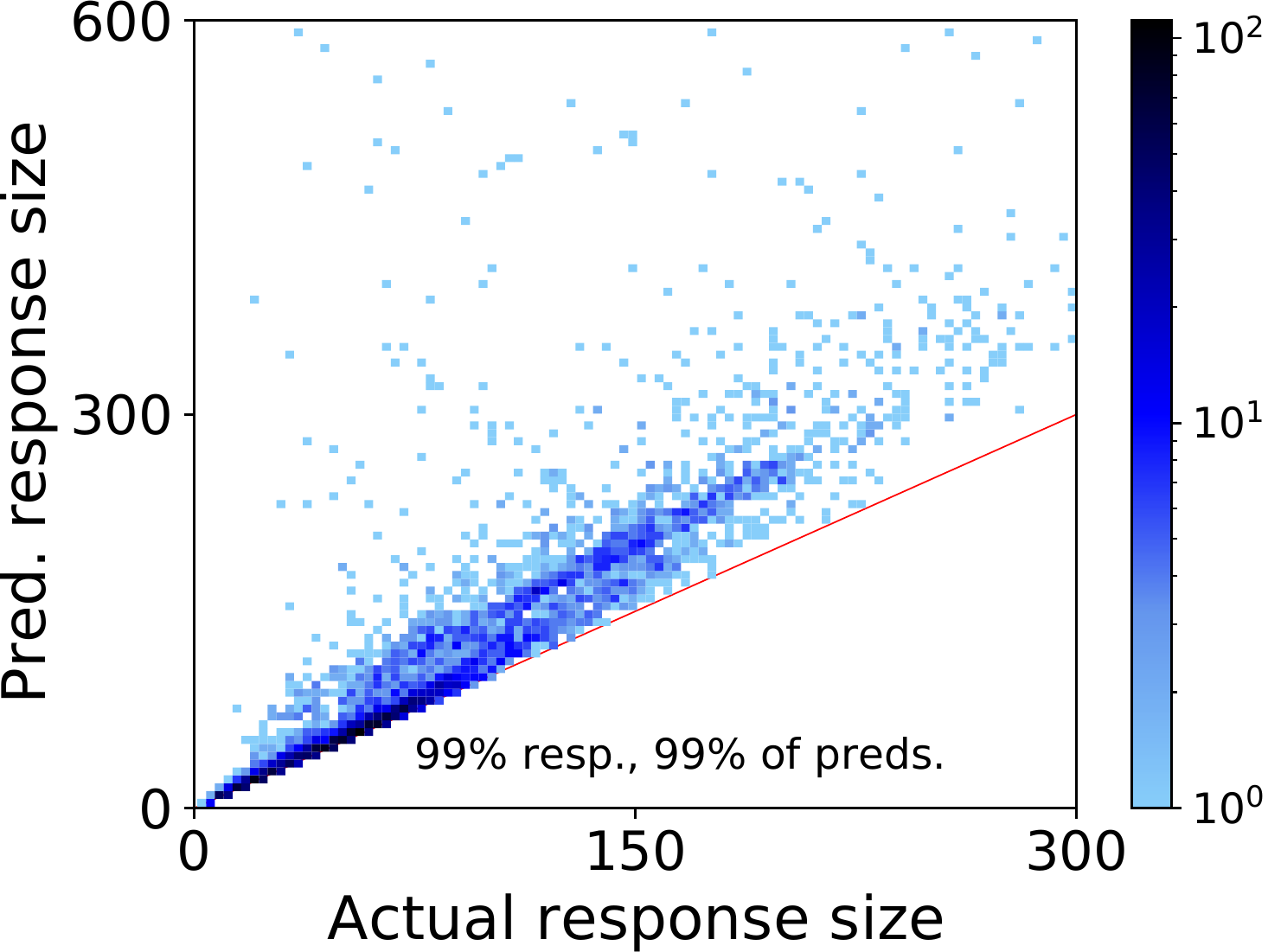}
  }
  \subfloat[libA's GitHub complexities]{
    \label{fig:OSSComparison_GitHub_libA}
    \includegraphics[width=0.50\columnwidth]{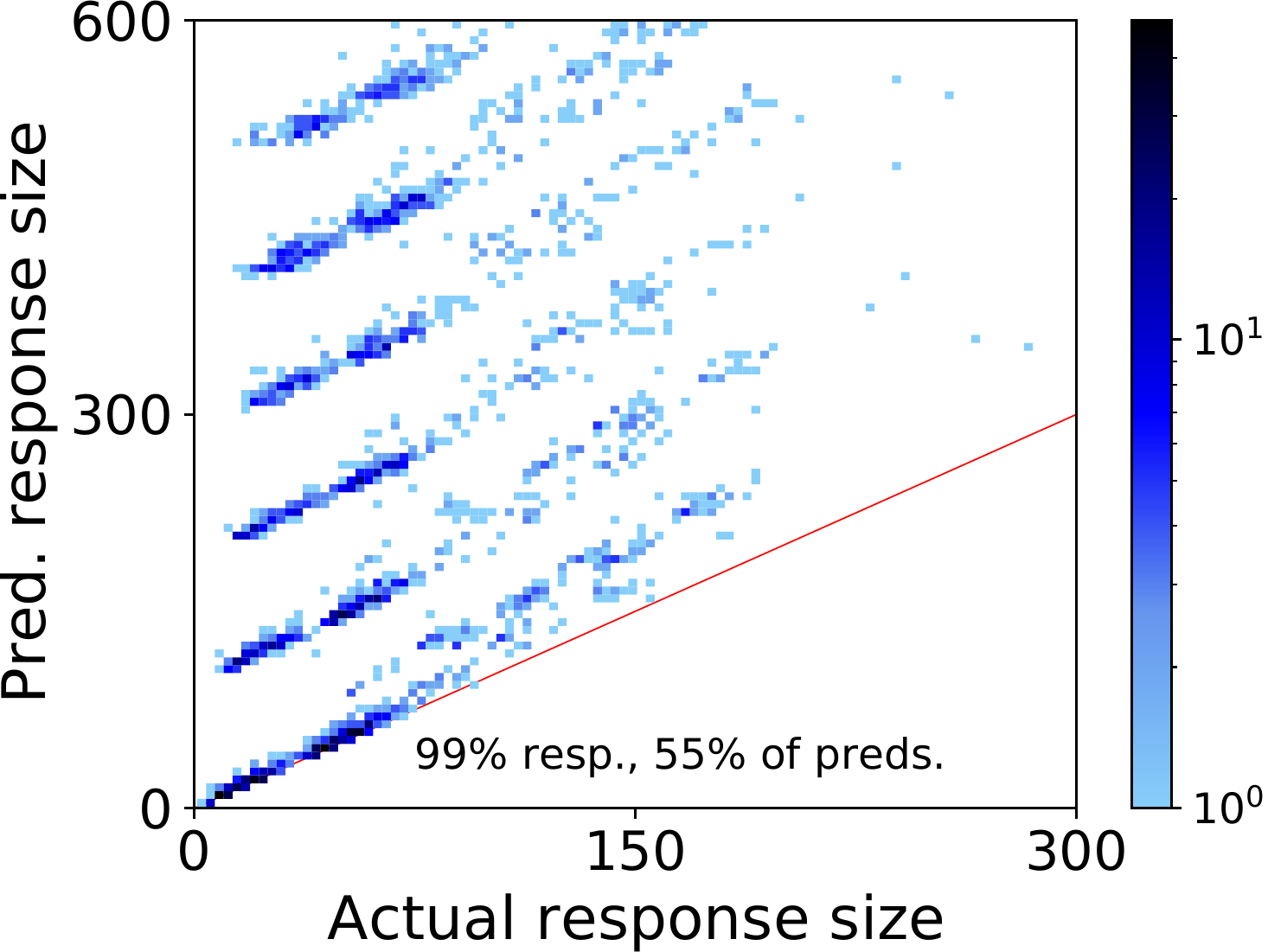}
  }
  \subfloat[libB's GitHub complexities]{
    \label{fig:OSSComparison_GitHub_libB}
    \includegraphics[width=0.50\columnwidth]{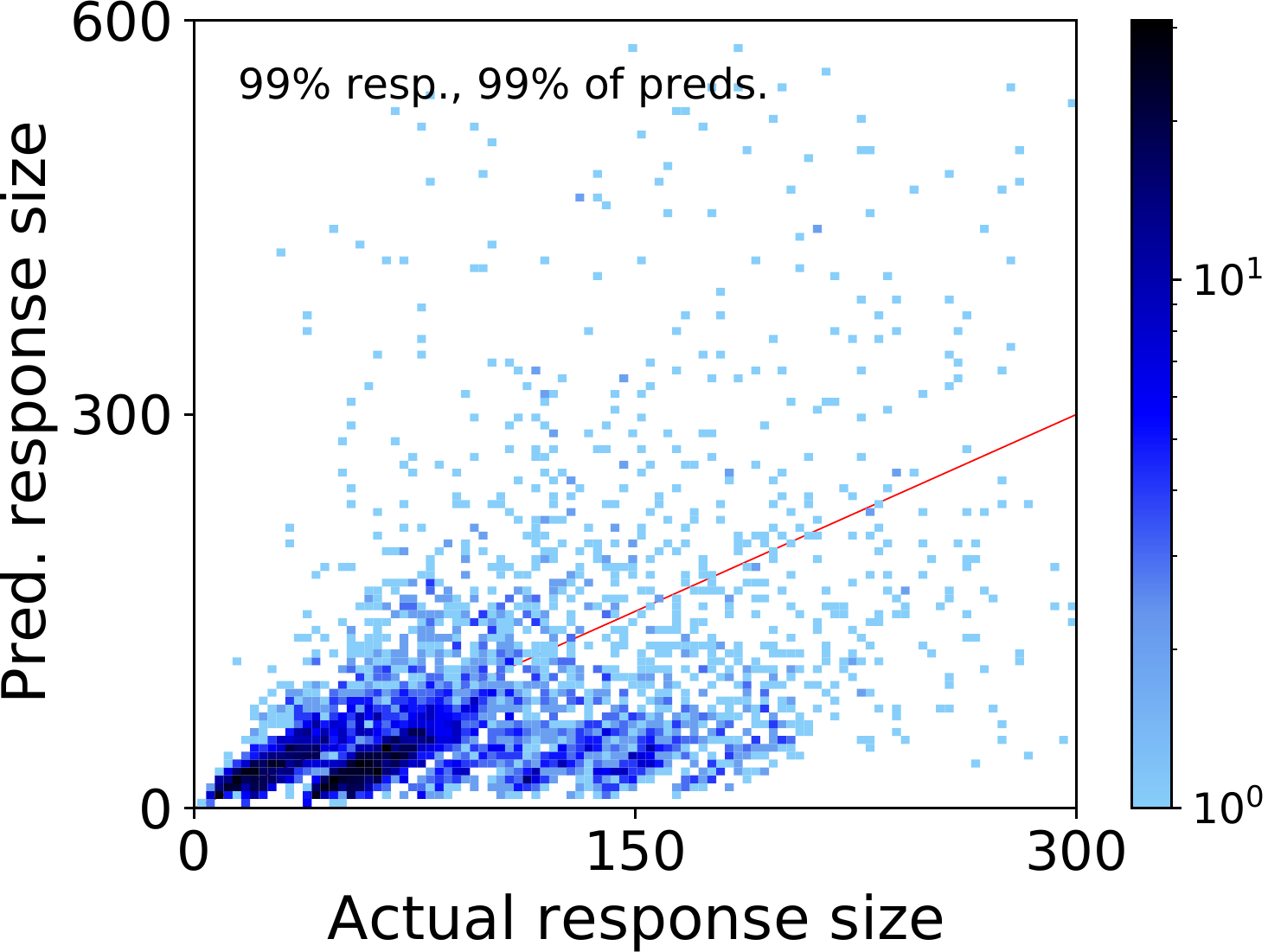}
  }
  \subfloat[libC's GitHub complexities]{
    \label{fig:OSSComparison_GitHub_libC}
    \includegraphics[width=0.50\columnwidth]{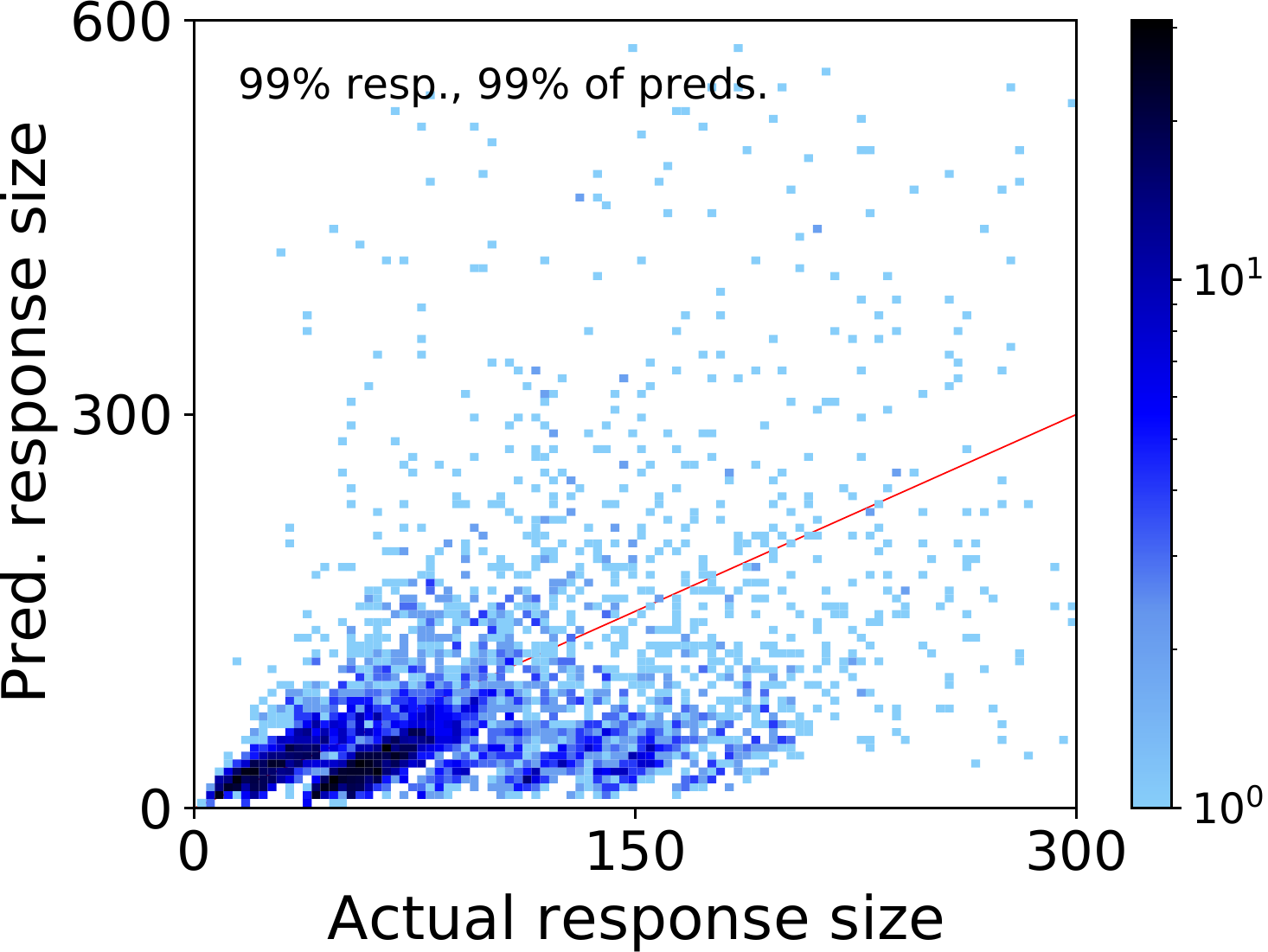}
  }
  \caption{
    \BQOne: Actual and predicted response sizes based on type complexity,
    from our analysis and the libraries on the GitHub data.
    \LibSlicknodeGQC and \LibPaBruGCA produce identical values under our configuration.
    All static approaches over-estimate due to data sparsity.
    The libraries have sources of over-estimation beyond our own.
    \LibSlicknodeGQC and \LibPaBruGCA also under-estimate (cells below the diagonal).
  }
  \label{fig:OSSComparison_GitHub}
\end{figure*}
}

\subsubsection{Configuring our Analysis to Answer \BQOne and \BQTwo} \label{sec:evaluation_comparison_Us}
Our measure of a query's type complexity can answer \BQOne.
Using the $W_{1,0}$ configuration (cf.~\cref{sec:static_analysis}),
the type complexity measures the maximum number of objects that can appear in a response, \ie response size.

Our measure of a query's resolve complexity is suitable for answering \BQTwo.
We assume that the cost of a resolver for a scalar or an enum field is paid for by some higher-level resolver.
Thus, we again configured our analysis for GitHub using the $W_{1,0}$ configuration from~\cref{sec:static_analysis}.
The resolve complexity resulting from this configuration will count each possible resolver function execution once.

\subsubsection{Comparison to Open-Source Static Analyses} \label{sec:evaluation_comparison_OpenSource}
We selected open-source libraries for comparison.
We describe their approaches in terms of our complexity measures,
configure them to answer the questions as best as possible,
and discuss their shortcomings.

\MyParagraph{Library selection.}
We considered analyses that met three criteria:
  (1) \textit{Existence}: They were hosted on GitHub, discoverable via searches on  
      {\code{graphql \{static|cost|query\} \{cost|analysis|complexity\}}}; \\
  (2) \textit{Relevance}: They statically compute a query cost measure;
  and
  (3) \textit{Quality}: They had $\geq 10$ stars, a rough proxy for quality~\cite{borges2018star}.

Three libraries met this criteria:
  \LibCatalyzerGVC~\cite{graphql-validation-complexity},
  \LibSlicknodeGQC~\cite{graphql-query-complexity},
  and
  \LibPaBruGCA~\cite{graphql-cost-analysis}.

\MyParagraph{Understanding their complexity measures.}
At a high level, each of these libraries defines a cost measure in terms of the data returned by each resolver function invoked to generate a response.
Users must specify the \textit{cost} of the entity type that it returns,
  and a \textit{multiplier} corresponding to the number of entities.
These libraries then compute a weighted recursive sum, much as we do.
Problematically, these libraries do not always permit users to specify the sizes of lists, leading to over- and under-estimates.
In terms of the analysis from~\cref{sec:static_analysis}, the list field size $l$ cannot always be obtained (see \code{qtcx}, last rule above Theorem 2).
We discuss the details below.


\MyParagraph{Configuring the libraries to answer \BQOne.}
We summarize our approach here.
Our artifact includes the configuration details~\cite{artifact}.

\renewcommand{\MyParagraph}[1]{\textit{#1}\xspace}
\renewcommand{\MyParagraph}[1]{\vspace{0.25em} \textit{#1}\xspace}
\renewcommand{\MyParagraph}[1]{\smallskip \textit{#1}\xspace}

\textbf{\LibCatalyzerGVC}:
We set the cost of all objects to $1$ and all other types to $0$.
The library allows the maximum size of a list to be set.
Therefore, we set the maximum sizes of unpaginated lists using values identified in~\cref{sec:evaluation_configuration}.
Unfortunately, the library cannot use arguments and consequently, cannot take advantage of pagination.
To configure paginated lists, we instead set their maximum sizes to GitHub's maximum list size, $100$~\cite{GitHubGraphQL:ResourceLimitations:2019}.


\textbf{\LibSlicknodeGQC}:
As with \LibCatalyzerGVC, we set the cost of objects to $1$ and other types to $0$.
We also leveraged \LibSlicknodeGQC's partial pagination support.
We configured it to use the limit arguments of paginated lists. 
However, \LibSlicknodeGQC treats unpaginated lists are treated as a single (non-list) entity.
It also does not support the connections pattern.


\textbf{\LibPaBruGCA}:
\LibSlicknodeGQC and \LibPaBruGCA appear to be related.
Their support for \BQOne is equivalent, and we configured them similarly.



\MyParagraph{Comparing outcomes on \BQOne.}
\cref{fig:OSSComparison_GitHub} illustrates the effectiveness of each approach when answering \BQOne.
The same query-response corpus was used in each case, and the response sizes were calculated using the method discussed in~\cref{sec:query_analysis}.
The variation is on the y-axis, the predicted response size estimated from the query.
As implied by our proofs of correctness, our approach consistently provides an upper bound on the actual response size (no under-estimation).

As mentioned, \LibCatalyzerGVC cannot use arguments as limits.
Configuring all paginated lists to have a maximum list size equal to GitHub's maximum list size resulted in significant over-estimation.
Additionally, this configuration caused in the striations in~\cref{fig:OSSComparison_GitHub_libA}, which lie at intervals of 100 (the maximum list size); queries fall into bands based on the number of paginated lists they contain.
In contrast to the striations found in~\cref{fig:OUR_complexity_heatmap_resolve_yelp}, these fields nest within each other, allowing for repeated and regularly spaced stripes.
To illustrate the overestimation, about 81.2\% of \LibCatalyzerGVC's predictions were more than double the actual response size.
In contrast, 84\% of our analysis's predictions over-estimate less than 50\% (cf. \cref{table:overestimation}).
Furthemore, the median over-estimation of our analysis (i.e. median error) is 19\%, whereas that of \LibCatalyzerGVC is 634\%. The 90\% percentile over-estimation of our analysis is 66.0\% and that of \LibCatalyzerGVC is 14,568.2\%.
As a result, our analysis conclusively performs better than \LibCatalyzerGVC.

\LibSlicknodeGQC and \LibPaBruGCA both over-estimate \emph{and} under-estimate.
Because they do not support default limits and treat unpaginated lists as a single (non-list) entity, they are prone to under-estimation.
The under-estimation can compound geometrically when multiple unpaginated lists are nested.
About 64\% of their predictions were under-estimations.
Additionally, because they do not support the connections pattern style of pagination and do not properly utilize limit arguments in these cases, they are prone to over-estimation.
In any case, because \LibSlicknodeGQC and \LibPaBruGCA can under-estimate, they do not reliably produce upper bounds, which is a problem in security critical contexts.
In contrast, our analysis consistently produces upper bounds and notably, tight upper bounds.
Because our analysis is not at risk of this security problem, our analysis also performs better than \LibSlicknodeGQC and \LibPaBruGCA.



\MyParagraph{Configuring the libraries to answer \BQTwo.}
We were \textit{unable} to configure these libraries to answer \BQTwo.
Fundamentally, \LibCatalyzerGVC, \LibSlicknodeGQC, and \LibPaBruGCA measure costs in terms of the entities returned by the query, not in terms of the resolvers used to obtain the results.
Trying to apply the multipliers to count resolvers will instead inflate them.
The trouble is illuminated by our resolve complexity analysis (\cref{fig:complexity}):
  in the first clause of the final rule, the resolver multiplier should be used to account for the $l$ uses of each child resolver, but the parent should be counted just once.
In contrast, when answering \BQOne, the type multiplier should be applied to both a field and its children.

This finding highlights the novelty of our notion of resolve complexity. 
We believe its ability to answer \BQTwo also shows its utility.

\subsubsection{Comparison to Closed-Source Analyses} \label{sec:evaluation_comparison_ClosedSource}
GitHub and Yelp describe their analyses in enough detail for comparison.
GitHub's analysis can approximate \BQOne and \BQTwo, while Yelp's cannot.

\MyParagraph{GitHub.}
GitHub's GraphQL API relies on two static analyses for rate limiting and blocking overly complex queries prior to executing them~\cite{GitHubGraphQL:ResourceLimitations:2019}.
Both analyses support pagination via the connections pattern as described in~\cref{sec:pagination}.
For \BQOne, their \emph{node limit} analysis disregards types associated with the connections pattern.\footnote{A possible explanation: virtual constructs like \gql{Edge} and \gql{Connection} types may not significantly increase the amount of effort to fulfill a query.}
We can replicate this behavior with our analysis by setting the weights of types associated with the connections pattern to~$0$ and~$1$ otherwise.
For \BQTwo, their \emph{call's score} analysis only counts resolvers that return \gql{Connection} types. 
We can also replicate this behavior by setting a weight of~$1$ to these resolvers and~$0$ otherwise.
In any case, because GitHub's metrics cannot be weighted, they cannot distinguish between more or less costly types or resolvers.

GitHub's focus on the connections pattern may have caused them issues in the past. 
When our study began, GitHub shared a shortcoming with \LibSlicknodeGQC and \LibPaBruGCA: it did not properly handle unpaginated lists (which would not employ the connections pattern).
We demonstrated the failure of their approach in~\cref{fig:architecture}.
We reported this possible denial of service vector to GitHub.
They confirmed the issue and have since patched their analysis.

\MyParagraph{Yelp.}
\label{sec:evaluation_discussion_yelp}
Yelp's GraphQL API analysis has both static and dynamic components.
Statically, Yelp's GraphQL API rejects queries with more than four levels of nesting.
This strategy bounds the complexity of valid queries, but expensive queries can still be constructed with this restriction using large nested lists. 
Dynamically, Yelp then applies rate limits by executing queries and retroactively replenishing the client's remaining rate limits according to the complexity of the response.
It is therefore possible to significantly exceed Yelp's rate limits by submitting a complex query when a client has a small quota remaining.
Using our type complexity analysis, Yelp could address this problem by rejecting queries whose estimated complexities exceeded a client's remaining rates.

\renewcommand{\MyParagraph}[1]{\vspace{0.18em}\textit{#1}\xspace}
\section{Discussion and Related Work}

%

\MyParagraph{Configuration and applicability.}
Our experiments show that our analysis is configurable to work with two real-world GraphQL APIs.
Applying our analysis was possible because it is static, \ie it does \emph{not} depend on any interaction with the GraphQL APIs or other backend systems.
This contrasts with dynamic analyses, which depend on probing backends for list sizes~\cite{Hartig:2018}.
Our analysis is more broadly applicable, and can be deployed separately from the GraphQL backend if desired, \eg in API gateways (cf.~\cref{sec:application}).
The static approach carries greater risk of over-estimation, however, and API providers may consider a hybrid approach similar to GitHub's: a static filter, then dynamic monitoring.


We have identified three strategies for managing over-estimation.
First, an unpaginated list field may produce responses with a wide range of sizes, leading our approach to overestimate. 
Schema designers may respond by paginating the list, which will bound the degree of overestimation. 
Second, in our tests we used the $W_{1,0}$ configuration, which assigned all types and resolvers the same weights.
In contexts where different data and resolvers carry different costs, schema designers can tune the configuration appropriately.
Lastly, service providers may resort to a hybrid static/dynamic system to leverage the graph data at runtime.
The design of such a system is a topic for further research.


\MyParagraph{The value of formalization.}
Our formal analysis gives us provably correct bounds, provided that list sizes can be obtained from an analysis configuration.
This contrasts with the more ad hoc approaches favored in the current generation of GraphQL analyses used by practitioners.
A formal approach ensured that we did not miss ``corner cases'', as in the unpaginated list entities missed by \LibSlicknodeGQC, \LibPaBruGCA, and GitHub's internal analysis.
Although our formalisms are not particularly complex, they guarantee the soundness and correctness missing from the state of the art.


\MyParagraph{Data-driven software engineering.}
Our approach benefited from an understanding of GraphQL as it is used in practice, specifically the use of pagination and naming conventions.
Although pagination is not part of the GraphQL specification~\cite{GraphQLInvention:2015},
  we found that the GraphQL grey literature emphasized the importance of pagination.
A recent empirical study of GraphQL schemas confirmed that various pagination strategies are widely used by practitioners~\cite{Wittern:2019}.
We therefore
  incorporated pagination into our formalization (viz. that list sizes can be derived from the context object)
  and
  supported both of the widely used pagination patterns in our configuration.
This decision differentiates our analysis from the state of the art, enabling us to avoid common sources of cost under- and over-estimation.
In addition, the prevalence of naming conventions in GraphQL schemas inspired our support for regular expressions, which allowed our configuration answer \BQOne and \BQTwo remarkably concisely.
In contrast, the libraries we used required us to manually specify costs and multipliers for each of the (hundreds of) GitHub schema elements --- they did not scale well to real-world schemas.

\MyParagraph{Bug finding.}
One surprising application of our analysis was as a bug finding tool.
When we configured Yelp's API, we assumed that limit arguments would be honored (\cref{hyp:resolve-limit}, \cref{sec:background_complexity}).
In early experiments we found that Yelp's resolver functions for the \gql{Query.reviews} and \gql{Business.reviews} fields ignore the limit argument.
Yelp's engineering team confirmed this to be a bug.
This interaction emphasized the validity of our assumptions.

\MyParagraph{Database Query Analyses.}
There has been significant work on query cost estimation for database query languages to optimize execution.
Our analysis is related to the estimation of a database query's cardinality and cost~\cite{OracleDB}. 
However, typical SQL servers routinely optimize queries by reordering table accesses, which makes static cost evaluation challenging~\cite{SQLEstimation}. 
In comparison, our analysis takes advantage of the limited expressivity of GraphQL, and the information provided in the schema (e.g., via pagination mechanism) to guarantee robust and precise upper-bounds before execution. 

\balance
\section{Application example: API Gateway}
\label{sec:application}
We designed our GraphQL query cost analysis as a building block for a GraphQL \emph{API gateway} that offers \emph{API management} for GraphQL backends (\cref{fig:architecture}). 
We worked with IBM's product division to implement a GraphQL API gateway based on our ideas.
Following patterns for API gateways for REST-like APIs~\cite{3Scale,APIConnect,APIGee}, this gateway is backend agnostic, made possible by our data- and backend-independent query cost analysis. 
This gateway was incorporated into v10.0.0 of IBM's API Connect and DataPower products~\cite{APIConnect2020}.

A weakness of our approach is the need for configuration.
During productization we explored two ways to support this task: a graphical user interface (GUI) and automatic recommendations.
The gateway automatically ingests the backend's schema using introspection~\cite{GraphQLDocs-Introspection}.
Users can then configure using the GUI depicted in~\cref{fig:screenshot}. 
They can manually configure fields with weights, limit arguments, and/or default limits, either one at a time or bulk-apply to all types/fields matching a search.
To mitigate the security risks of schemas with nested structures, the GUI automatically identifies some problematic fields and proposes an appropriate configuration based on schema conventions.
For example, it flags fields that return lists and infers possible configurations based on type information. 

{
\setlength{\belowcaptionskip}{0em} 
\begin{figure}
  \centering
  \includegraphics[width=0.95\columnwidth]{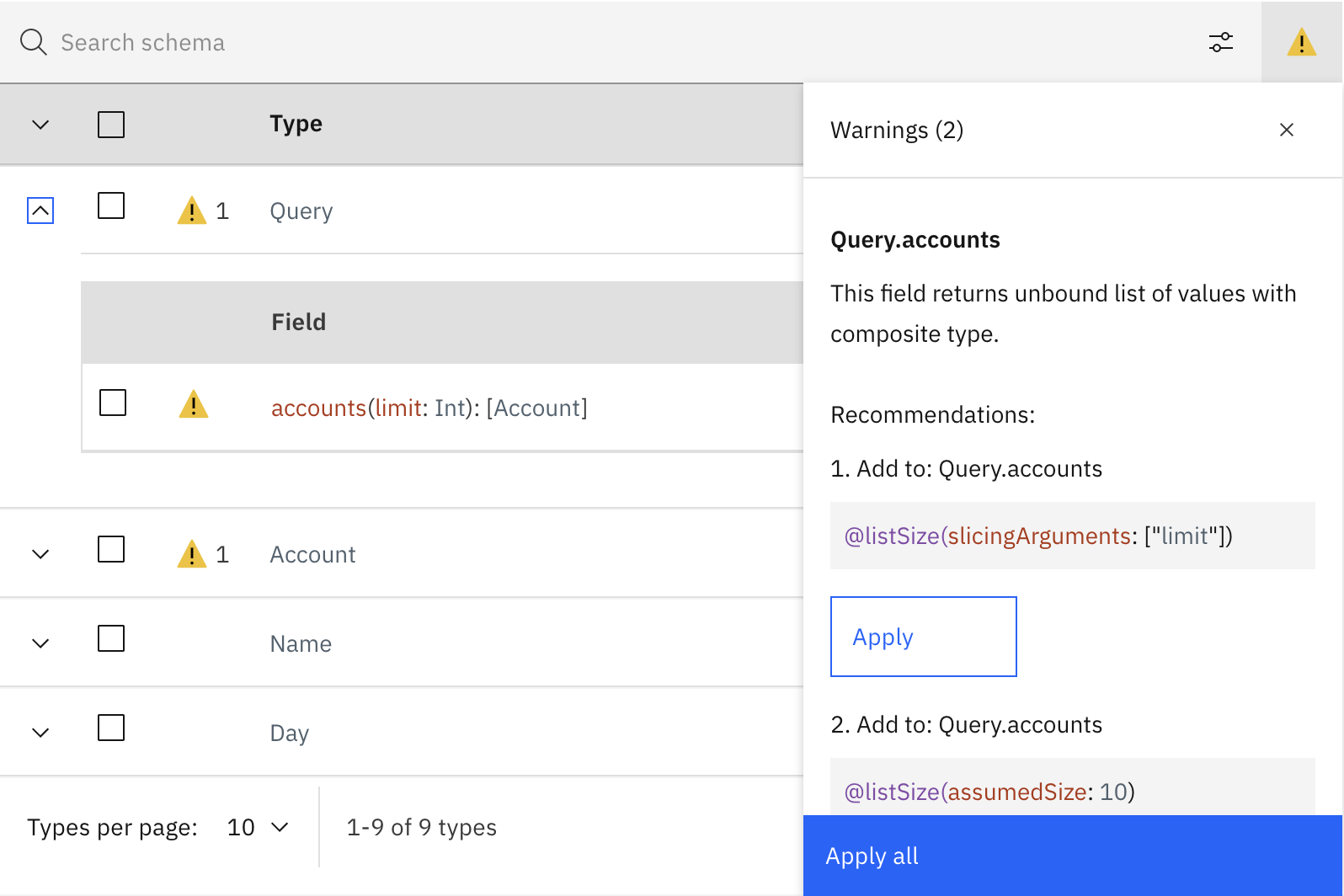}
  \caption{
    Screenshot of the configuration GUI in IBM's DataPower API gateway.
    The warnings indicate incomplete configuration and make recommendations.
  }
  \label{fig:screenshot}
\end{figure}
}



\section{Threats to Validity}
\label{sec:threats_to_validity}

\MyParagraph{Construct validity.}
Our study does not face significant threats to construct validity.
We believe our definitions of type complexity and resolve complexity are useful.
We do not rely on proxy measures, but rather measure these complexities directly from real queries.

\MyParagraph{Internal validity.}
These threats come from configuration and query realism.
In our evaluation, we created configurations for the GitHub and Yelp APIs.
Errors would affect the accuracy of our bounds.
Our evaluation showed that we did not make errors leading to under-estimation, but we may have done so for over-estimation.

Although RQ2 showed that our analysis produces upper bounds, in RQ3 our conclusions about the practicality of our bounds rely on the realism of our query-response corpus.
Our evaluation is based on randomly generated queries, parameterized as described in~\cref{sec:evaluation_query-generation}.
Some of our queries used popular projects, which are more likely to have associated data (decreasing over-estimates from data sparsity).
Other queries lacked contextual knowledge and may result in ``unnatural'' queries unlikely to be filled with data.
To avoid harming the public API providers, we bounded the complexity of the queries we issued, and this may have skewed our queries to be smaller than realistic queries.
We plan to pursue a more realistic set of queries, \eg obtained through collaboration with a GraphQL API provider or by mining queries from open-source software.

\MyParagraph{External validity.}
Our work makes assumptions about the properties of GraphQL schemas and backend implementations that may not hold for all GraphQL API providers.
For example, the complexity calculations depend on the presence of slicing arguments in queries, on resolver function implementations to enforce these limits, and on a proper configuration.
By relying on default limits~(\cref{sec:analysis_configuration}), we enable our analysis to function even if slicing arguments are not enforced in (parts of) a schema. 
We demonstrated that proper configuration is possible even when treating the backend as a grey box, as we did when evaluating on the GitHub and Yelp APIs (\cref{sec:evaluation_configuration}).





\section{Conclusion}
\label{sec:conclusion}
GraphQL is an emerging web API model.
Its flexibility can benefit clients, servers, and network operators.
But its flexibility is also a threat: GraphQL queries can be exponentially complex, with implications for service providers including rate limiting and denial of service.
The fundamental requirement for service providers is a cheap, accurate way to estimate the cost of a query.

We showed in our evaluation that existing \textit{ad hoc} approaches are liable to both over-estimates and under-estimates.
We proposed instead a principled approach to address this challenge.
Grounded in a formalization of GraphQL semantics, in this work we presented the first provably-correct static query cost analyses. 
With proper configuration, our analysis offers tight upper bounds, low runtime overhead, and independence from backend implementation details. 
We accompany our work with the first GraphQL query-response corpus to support future research.


\section*{Reproducibility}

An artifact containing
  the GraphQL query generator,
  the query-response corpuses,
  library configurations,
  and corpus measurements
  can be found here: \href{https://zenodo.org/record/4023299}{https://zenodo.org/record/4023299}.
Institutional policy precludes sharing our analysis prototype. 


\section*{Acknowledgments}

The authors are grateful to the reviewers and to A. Kazerouni for their helpful feedback.
We thank the IBM API Connect and DataPower teams for working with us on the GraphQL API Gateway.
\newpage

\balance
\bibliographystyle{ACM-Reference-Format}
\bibliography{bibliography}


\begin{thebibliography}{41}


\ifx \showCODEN    \undefined \def \showCODEN     #1{\unskip}     \fi
\ifx \showDOI      \undefined \def \showDOI       #1{#1}\fi
\ifx \showISBNx    \undefined \def \showISBNx     #1{\unskip}     \fi
\ifx \showISBNxiii \undefined \def \showISBNxiii  #1{\unskip}     \fi
\ifx \showISSN     \undefined \def \showISSN      #1{\unskip}     \fi
\ifx \showLCCN     \undefined \def \showLCCN      #1{\unskip}     \fi
\ifx \shownote     \undefined \def \shownote      #1{#1}          \fi
\ifx \showarticletitle \undefined \def \showarticletitle #1{#1}   \fi
\ifx \showURL      \undefined \def \showURL       {\relax}        \fi
\providecommand\bibfield[2]{#2}
\providecommand\bibinfo[2]{#2}
\providecommand\natexlab[1]{#1}
\providecommand\showeprint[2][]{arXiv:#2}

\bibitem[\protect\citeauthoryear{??}{Sta}{2016}]%
        {StackOverflow:2016}
 \bibinfo{year}{2016}\natexlab{}.
\newblock \bibinfo{title}{{How do you prevent nested attack on GraphQL/Apollo
  server?}}
\newblock
\newblock
\urldef\tempurl%
\url{https://web.archive.org/web/20200910231657/https://stackoverflow.com/questions/37337466/how-do-you-prevent-nested-attack-on-graphql-apollo-server/37338465}
\showURL{%
\tempurl}


\bibitem[\protect\citeauthoryear{??}{Yel}{2017}]%
        {YelpGraphQL:2017}
 \bibinfo{year}{2017}\natexlab{}.
\newblock \bibinfo{title}{{Yelp -- Introducing Yelp's Local Graph}}.
\newblock
\newblock
\urldef\tempurl%
\url{https://web.archive.org/web/20200910231907/https://engineeringblog.yelp.com/2017/05/introducing-yelps-local-graph.html}
\showURL{%
\tempurl}


\bibitem[\protect\citeauthoryear{??}{Gra}{2018}]%
        {GraphQLSpec}
 \bibinfo{year}{2018}\natexlab{}.
\newblock \bibinfo{title}{{GraphQL Specification}}.
\newblock
\newblock
\urldef\tempurl%
\url{https://graphql.org/graphql-spec/}
\showURL{%
\tempurl}


\bibitem[\protect\citeauthoryear{??}{Con}{2019}]%
        {Contentful:Ratelimits:2019}
 \bibinfo{year}{2019}\natexlab{}.
\newblock \bibinfo{title}{{Contentful -- Query complexity limits}}.
\newblock
\newblock
\urldef\tempurl%
\url{https://www.contentful.com/developers/docs/references/graphql/#/introduction/api-rate-limits}
\showURL{%
\tempurl}


\bibitem[\protect\citeauthoryear{??}{Git}{2019a}]%
        {GitHubGraphQL:2019}
 \bibinfo{year}{2019}\natexlab{a}.
\newblock \bibinfo{title}{{GitHub -- GraphQL API v4}}.
\newblock
\newblock
\urldef\tempurl%
\url{https://developer.github.com/v4/}
\showURL{%
\tempurl}


\bibitem[\protect\citeauthoryear{??}{Git}{2019b}]%
        {GitHubGraphQL:Examples:2019}
 \bibinfo{year}{2019}\natexlab{b}.
\newblock \bibinfo{title}{{GitHub -- GraphQL Example Queries}}.
\newblock
\newblock
\urldef\tempurl%
\url{https://github.com/github/platform-samples/tree/master/graphql/queri}
\showURL{%
\tempurl}


\bibitem[\protect\citeauthoryear{??}{Git}{2019c}]%
        {GitHubGraphQL:ResourceLimitations:2019}
 \bibinfo{year}{2019}\natexlab{c}.
\newblock \bibinfo{title}{{GitHub GraphQL API v4: GraphQL resource
  limitations}}.
\newblock
\newblock
\urldef\tempurl%
\url{https://developer.github.com/v4/guides/resource-limitations/}
\showURL{%
\tempurl}


\bibitem[\protect\citeauthoryear{??}{Gra}{2019a}]%
        {GraphiQL}
 \bibinfo{year}{2019}\natexlab{a}.
\newblock \bibinfo{title}{{GraphiQL -- An in-browser IDE for exploring
  GraphQL}}.
\newblock
\newblock
\urldef\tempurl%
\url{https://github.com/graphql/graphiql}
\showURL{%
\tempurl}


\bibitem[\protect\citeauthoryear{??}{Gra}{2019b}]%
        {GraphQLDocs-Introspection}
 \bibinfo{year}{2019}\natexlab{b}.
\newblock \bibinfo{title}{{GraphQL Docs: Introspection}}.
\newblock
\newblock
\urldef\tempurl%
\url{https://graphql.org/learn/introspection/}
\showURL{%
\tempurl}


\bibitem[\protect\citeauthoryear{??}{Gra}{2019c}]%
        {GraphQLDocs-Pagination}
 \bibinfo{year}{2019}\natexlab{c}.
\newblock \bibinfo{title}{{GraphQL Docs: Pagination}}.
\newblock
\newblock
\urldef\tempurl%
\url{http://graphql.org/learn/pagination/}
\showURL{%
\tempurl}


\bibitem[\protect\citeauthoryear{??}{Gra}{2019d}]%
        {GraphQLDocs-QueryAndMutationTypes}
 \bibinfo{year}{2019}\natexlab{d}.
\newblock \bibinfo{title}{{GraphQL Docs: The Query and Mutation types}}.
\newblock
\newblock
\urldef\tempurl%
\url{https://graphql.org/learn/schema/\#the-query-and-mutation-types}
\showURL{%
\tempurl}


\bibitem[\protect\citeauthoryear{??}{Gra}{2019e}]%
        {GraphQL-Faker}
 \bibinfo{year}{2019}\natexlab{e}.
\newblock \bibinfo{title}{{GraphQL Faker}}.
\newblock
\newblock
\urldef\tempurl%
\url{https://github.com/APIs-guru/graphql-faker}
\showURL{%
\tempurl}


\bibitem[\protect\citeauthoryear{??}{Gra}{2019f}]%
        {GraphQL-js}
 \bibinfo{year}{2019}\natexlab{f}.
\newblock \bibinfo{title}{{GraphQL.js -- JavaScript reference implementation
  for GraphQL}}.
\newblock
\newblock
\urldef\tempurl%
\url{https://github.com/graphql/graphql-js}
\showURL{%
\tempurl}


\bibitem[\protect\citeauthoryear{??}{Ora}{2019}]%
        {OracleDB}
 \bibinfo{year}{2019}\natexlab{}.
\newblock \bibinfo{title}{{Oracle Database Documentation}}.
\newblock
\newblock
\urldef\tempurl%
\url{https://docs.oracle.com/database}
\showURL{%
\tempurl}


\bibitem[\protect\citeauthoryear{??}{API}{2019}]%
        {APIsGuru-GraphQLAPIs}
 \bibinfo{year}{2019}\natexlab{}.
\newblock \bibinfo{title}{{Public GraphQL APIs}}.
\newblock
\newblock
\urldef\tempurl%
\url{https://github.com/APIs-guru/graphql-apis}
\showURL{%
\tempurl}


\bibitem[\protect\citeauthoryear{??}{Rel}{2019}]%
        {Relay-Pagination}
 \bibinfo{year}{2019}\natexlab{}.
\newblock \bibinfo{title}{{Relay -- Pagination Specification}}.
\newblock
\newblock
\urldef\tempurl%
\url{https://facebook.github.io/relay/graphql/connections.htm}
\showURL{%
\tempurl}


\bibitem[\protect\citeauthoryear{??}{Sho}{2019a}]%
        {ShopifyGraphQL:Ratelimits:2019}
 \bibinfo{year}{2019}\natexlab{a}.
\newblock \bibinfo{title}{{Shopify -- GraphQL Admin API rate limits}}.
\newblock
\newblock
\urldef\tempurl%
\url{https://shopify.dev/concepts/about-apis/rate-limits#graphql-admin-api-rate-limits}
\showURL{%
\tempurl}


\bibitem[\protect\citeauthoryear{??}{Sho}{2019b}]%
        {ShopifyGraphQL:2019}
 \bibinfo{year}{2019}\natexlab{b}.
\newblock \bibinfo{title}{{Shopify -- Shopify Storefront API}}.
\newblock
\newblock
\urldef\tempurl%
\url{https://shopify.dev/docs/storefront-api}
\showURL{%
\tempurl}


\bibitem[\protect\citeauthoryear{??}{Yel}{2019}]%
        {Yelp:Ratelimits:2019}
 \bibinfo{year}{2019}\natexlab{}.
\newblock \bibinfo{title}{{Yelp -- GraphQL API Points-Based Daily Limit}}.
\newblock
\newblock
\urldef\tempurl%
\url{https://www.yelp.com/developers/graphql/guides/rate_limiting}
\showURL{%
\tempurl}


\bibitem[\protect\citeauthoryear{??}{gra}{2020a}]%
        {graphql-validation-complexity}
 \bibinfo{year}{2020}\natexlab{a}.
\newblock \bibinfo{title}{{4Catalyzer/graphql-validation-complexity: Query
  complexity validation for GraphQL.js}}.
\newblock
\newblock
\urldef\tempurl%
\url{https://github.com/4Catalyzer/graphql-validation-complexity}
\showURL{%
\tempurl}


\bibitem[\protect\citeauthoryear{??}{API}{2020a}]%
        {APIGee}
 \bibinfo{year}{2020}\natexlab{a}.
\newblock \bibinfo{title}{{Google Apigee}}.
\newblock
\newblock
\urldef\tempurl%
\url{https://cloud.google.com/apigee/}
\showURL{%
\tempurl}


\bibitem[\protect\citeauthoryear{??}{API}{2020b}]%
        {APIConnect}
 \bibinfo{year}{2020}\natexlab{b}.
\newblock \bibinfo{title}{{IBM API Connect}}.
\newblock
\newblock
\urldef\tempurl%
\url{https://www.ibm.com/cloud/api-connect}
\showURL{%
\tempurl}


\bibitem[\protect\citeauthoryear{??}{gra}{2020b}]%
        {graphql-cost-analysis}
 \bibinfo{year}{2020}\natexlab{b}.
\newblock \bibinfo{title}{{pabru/graphql-cost-analysis: A Graphql query cost
  analyzer}}.
\newblock
\newblock
\urldef\tempurl%
\url{https://github.com/pa-bru/graphql-cost-analysis}
\showURL{%
\tempurl}


\bibitem[\protect\citeauthoryear{??}{3Sc}{2020}]%
        {3Scale}
 \bibinfo{year}{2020}\natexlab{}.
\newblock \bibinfo{title}{{RedHat 3Scale}}.
\newblock
\newblock
\urldef\tempurl%
\url{https://www.3scale.net/}
\showURL{%
\tempurl}


\bibitem[\protect\citeauthoryear{??}{gra}{2020c}]%
        {graphql-query-complexity}
 \bibinfo{year}{2020}\natexlab{c}.
\newblock \bibinfo{title}{{slicknode/graphql-query-complexity: GraphQL query
  complexity analysis and validation for graphql-js}}.
\newblock
\newblock
\urldef\tempurl%
\url{https://github.com/slicknode/graphql-query-complexity}
\showURL{%
\tempurl}


\bibitem[\protect\citeauthoryear{??}{Gra}{2020}]%
        {GraphQLUsers}
 \bibinfo{year}{2020}\natexlab{}.
\newblock \bibinfo{title}{{Who's using GraphQL?}}
\newblock
\newblock
\urldef\tempurl%
\url{http://graphql.org/users}
\showURL{%
\tempurl}


\bibitem[\protect\citeauthoryear{Andersson}{Andersson}{2018}]%
        {Andersson2018GQLResultSizeCalculation}
\bibfield{author}{\bibinfo{person}{Tim Andersson}.}
  \bibinfo{year}{2018}\natexlab{}.
\newblock \emph{\bibinfo{title}{{Result size calculation for Facebook's GraphQL
  query language}}}.
\newblock {B.S. Thesis}.
\newblock
\urldef\tempurl%
\url{http://www.diva-portal.org/smash/get/diva2:1237221/FULLTEXT01.pdf}
\showURL{%
\tempurl}


\bibitem[\protect\citeauthoryear{Borges and Valente}{Borges and
  Valente}{2018}]%
        {borges2018star}
\bibfield{author}{\bibinfo{person}{Hudson Borges} {and}
  \bibinfo{person}{Marco~Tulio Valente}.} \bibinfo{year}{2018}\natexlab{}.
\newblock \bibinfo{title}{What's in a GitHub star? understanding repository
  starring practices in a social coding platform}.
\newblock , \bibinfo{numpages}{112--129}~pages.
\newblock


\bibitem[\protect\citeauthoryear{Brito, Mombach, and Valente}{Brito
  et~al\mbox{.}}{2019}]%
        {Brito:2019}
\bibfield{author}{\bibinfo{person}{Gleison Brito}, \bibinfo{person}{Thais
  Mombach}, {and} \bibinfo{person}{Marco~Tulio Valente}.}
  \bibinfo{year}{2019}\natexlab{}.
\newblock \showarticletitle{{Migrating to GraphQL: A Practical Assessment}}. In
  \bibinfo{booktitle}{\emph{2019 IEEE 26th International Conference on Software
  Analysis, Evolution and Reengineering (SANER)}}. IEEE,
  \bibinfo{pages}{140--150}.
\newblock
\urldef\tempurl%
\url{https://doi.org/10.1109/SANER.2019.8667986}
\showDOI{\tempurl}


\bibitem[\protect\citeauthoryear{Byron}{Byron}{2015}]%
        {GraphQLInvention:2015}
\bibfield{author}{\bibinfo{person}{Lee Byron}.}
  \bibinfo{year}{2015}\natexlab{}.
\newblock \bibinfo{title}{{GraphQL: A data query language}}.
\newblock
\newblock
\urldef\tempurl%
\url{https://web.archive.org/web/20200910232048/https://engineering.fb.com/core-data/graphql-a-data-query-language/}
\showURL{%
\tempurl}


\bibitem[\protect\citeauthoryear{Cha, Wittern, Baudart, Davis, Mandel, and
  Laredo}{Cha et~al\mbox{.}}{2020}]%
        {artifact}
\bibfield{author}{\bibinfo{person}{Alan Cha}, \bibinfo{person}{Erik Wittern},
  \bibinfo{person}{Guillaume Baudart}, \bibinfo{person}{James~C. Davis},
  \bibinfo{person}{Louis Mandel}, {and} \bibinfo{person}{Jim~A. Laredo}.}
  \bibinfo{year}{2020}\natexlab{}.
\newblock \bibinfo{title}{{A Principled Approach to GraphQL Query Cost Analysis
  Research Paper Artifact}}.
\newblock
\newblock
\urldef\tempurl%
\url{https://doi.org/10.5281/zenodo.4023299}
\showDOI{\tempurl}


\bibitem[\protect\citeauthoryear{Crosby and Wallach}{Crosby and
  Wallach}{2003}]%
        {Crosby:2003}
\bibfield{author}{\bibinfo{person}{Scott~A. Crosby} {and}
  \bibinfo{person}{Dan~S. Wallach}.} \bibinfo{year}{2003}\natexlab{}.
\newblock \showarticletitle{{Denial of Service via Algorithmic Complexity
  Attacks}}. In \bibinfo{booktitle}{\emph{Proceedings of the 12th Conference on
  USENIX Security Symposium - Volume 12}} (Washington, DC)
  \emph{(\bibinfo{series}{SSYM'03})}. \bibinfo{publisher}{USENIX Association},
  \bibinfo{pages}{29--44}.
\newblock


\bibitem[\protect\citeauthoryear{Godefroid, Kiezun, and Levin}{Godefroid
  et~al\mbox{.}}{2008}]%
        {godefroid2008grammar}
\bibfield{author}{\bibinfo{person}{Patrice Godefroid}, \bibinfo{person}{Adam
  Kiezun}, {and} \bibinfo{person}{Michael~Y. Levin}.}
  \bibinfo{year}{2008}\natexlab{}.
\newblock \showarticletitle{Grammar-based whitebox fuzzing}. In
  \bibinfo{booktitle}{\emph{Proceedings of the 29th ACM SIGPLAN Conference on
  Programming Language Design and Implementation}}. \bibinfo{pages}{206--215}.
\newblock
\urldef\tempurl%
\url{https://doi.org/10.1145/1375581.1375607}
\showDOI{\tempurl}


\bibitem[\protect\citeauthoryear{Hartig and P{\'e}rez}{Hartig and
  P{\'e}rez}{2018}]%
        {Hartig:2018}
\bibfield{author}{\bibinfo{person}{Olaf Hartig} {and} \bibinfo{person}{Jorge
  P{\'e}rez}.} \bibinfo{year}{2018}\natexlab{}.
\newblock \showarticletitle{{Semantics and Complexity of GraphQL}}. In
  \bibinfo{booktitle}{\emph{Proceedings of the 2018 World Wide Web Conference}}
  (Lyon, France) \emph{(\bibinfo{series}{WWW '18})}.
  \bibinfo{publisher}{International World Wide Web Conferences Steering
  Committee}, \bibinfo{address}{Republic and Canton of Geneva, Switzerland},
  \bibinfo{pages}{1155--1164}.
\newblock
\showISBNx{978-1-4503-5639-8}
\urldef\tempurl%
\url{https://doi.org/10.1145/3178876.3186014}
\showDOI{\tempurl}


\bibitem[\protect\citeauthoryear{Li, K{\"{o}}nig, Narasayya, and Chaudhuri}{Li
  et~al\mbox{.}}{2012}]%
        {SQLEstimation}
\bibfield{author}{\bibinfo{person}{Jiexing Li}, \bibinfo{person}{Arnd~Christian
  K{\"{o}}nig}, \bibinfo{person}{Vivek~R. Narasayya}, {and}
  \bibinfo{person}{Surajit Chaudhuri}.} \bibinfo{year}{2012}\natexlab{}.
\newblock , \bibinfo{numpages}{1555--1566}~pages.
\newblock


\bibitem[\protect\citeauthoryear{Rinquin}{Rinquin}{2017}]%
        {Rinquin:2017}
\bibfield{author}{\bibinfo{person}{Arnaud Rinquin}.}
  \bibinfo{year}{2017}\natexlab{}.
\newblock \bibinfo{title}{{Avoiding n+1 requests in GraphQL, including within
  subscriptions}}.
\newblock
\newblock
\urldef\tempurl%
\url{https://web.archive.org/web/20200910232552/https://medium.com/slite/avoiding-n-1-requests-in-graphql-including-within-subscriptions-f9d7867a257d}
\showURL{%
\tempurl}


\bibitem[\protect\citeauthoryear{Shrock}{Shrock}{2015}]%
        {2015GraphQLMotivation}
\bibfield{author}{\bibinfo{person}{Nick Shrock}.}
  \bibinfo{year}{2015}\natexlab{}.
\newblock \bibinfo{title}{{GraphQL Introduction}}.
\newblock
\newblock
\urldef\tempurl%
\url{https://web.archive.org/web/20200414211542/https://reactjs.org/blog/2015/05/01/graphql-introduction.html}
\showURL{%
\tempurl}


\bibitem[\protect\citeauthoryear{Stoiber}{Stoiber}{2018}]%
        {Stoiber:2018}
\bibfield{author}{\bibinfo{person}{Max Stoiber}.}
  \bibinfo{year}{2018}\natexlab{}.
\newblock \bibinfo{title}{{Securing Your GraphQL API from Malicious Queries}}.
\newblock
\newblock
\urldef\tempurl%
\url{https://web.archive.org/web/20200910232751/https://www.apollographql.com/blog/securing-your-graphql-api-from-malicious-queries-16130a324a6b/}
\showURL{%
\tempurl}


\bibitem[\protect\citeauthoryear{Thelen}{Thelen}{2020}]%
        {APIConnect2020}
\bibfield{author}{\bibinfo{person}{Rob Thelen}.}
  \bibinfo{year}{2020}\natexlab{}.
\newblock \bibinfo{title}{{API Connect is making GraphQL safer for the
  enterprise}}.
\newblock
\newblock
\urldef\tempurl%
\url{https://web.archive.org/web/20200910232932/https://community.ibm.com/community/user/imwuc/blogs/rob-thelen1/2020/06/16/api-connect-is-making-graphql-safer-for-the-enterp}
\showURL{%
\tempurl}


\bibitem[\protect\citeauthoryear{Wittern, Cha, Davis, Baudart, and
  Mandel}{Wittern et~al\mbox{.}}{2019}]%
        {Wittern:2019}
\bibfield{author}{\bibinfo{person}{Erik Wittern}, \bibinfo{person}{Alan Cha},
  \bibinfo{person}{James~C. Davis}, \bibinfo{person}{Guillaume Baudart}, {and}
  \bibinfo{person}{Louis Mandel}.} \bibinfo{year}{2019}\natexlab{}.
\newblock \showarticletitle{An Empirical Study of GraphQL Schemas}. In
  \bibinfo{booktitle}{\emph{Proceedings of the 17th International Conference on
  Service-Oriented Computing (ICSOC)}}, Vol.~\bibinfo{volume}{11895}.
\newblock


\bibitem[\protect\citeauthoryear{Wittern, Cha, and Laredo}{Wittern
  et~al\mbox{.}}{2018}]%
        {Wittern:2018}
\bibfield{author}{\bibinfo{person}{Erik Wittern}, \bibinfo{person}{Alan Cha},
  {and} \bibinfo{person}{Jim~A. Laredo}.} \bibinfo{year}{2018}\natexlab{}.
\newblock \showarticletitle{{Generating GraphQL-Wrappers for REST (-like)
  APIs}}. In \bibinfo{booktitle}{\emph{International Conference on Web
  Engineering (ICWE '18)}}. Springer International Publishing,
  \bibinfo{pages}{65--83}.
\newblock
\urldef\tempurl%
\url{https://doi.org/10.1007/978-3-319-91662-0_5}
\showDOI{\tempurl}


\end{thebibliography}

\clearpage
\end{document}